\documentclass[10pt]{article} 
\usepackage[a4paper]{geometry}
\usepackage{amsmath,amsthm,amssymb}

\newtheorem{theorem}{Theorem}[section]
\newtheorem{lemma}{Lemma}[section]

\newtheorem{remark}{Remark}[section]
\newtheorem{definition}{Definition}[section]

\newtheorem{example}{Example}[section]

\usepackage{rotating}

\usepackage{natbib}
\bibpunct[, ]{(}{)}{,}{a}{}{,}%

\usepackage[appendix = inline]{apxproof}
\newtheoremrep{theorem}{Theorem}[section]
\newtheoremrep{lemma}{Lemma}[section]
\newtheoremrep{corollary}{Corollary}[section]

\usepackage{color,hyperref}
\definecolor{darkblue}{rgb}{0.0,0.0,0.6}
\hypersetup{colorlinks,breaklinks,linkcolor=darkblue,urlcolor=darkblue,anchorcolor=darkblue,citecolor=darkblue}
\usepackage{algpseudocode}
\usepackage{algorithm}

\usepackage{algpseudocode}
\usepackage[bb=dsserif]{mathalpha}
\usepackage{bm}

\begin{document}
	
	
	
	
	

\title{Robust Trading in a Generalized Lattice Market}
\author{Chung-Han Hsieh\footnote{Department of Quantitative Finance, National Tsing Hua University, Hsinchu, Taiwan, 30004,  \href{mailto:ch.hsieh@mx.nthu.edu.tw}{ch.hsieh@mx.nthu.edu.tw}} \, and Xin-Yu Wang\footnote{Department of Quantitative Finance, National Tsing Hua University, Hsinchu, Taiwan, 30004, \href{mailto:xinyuwang@gapp.nthu.edu.tw}{xinyuwang@gapp.nthu.edu.tw}.}
}
 \date{}
\maketitle

\begin{abstract}
	This paper introduces a novel robust trading paradigm, called \textit{multi-double linear policies}, situated within a \textit{generalized} lattice market.
	Distinctively, our framework departs from most existing robust trading strategies, which are predominantly limited to single or paired assets and typically embed asset correlation within the trading strategy itself, rather than as an inherent characteristic of the market. 
	Our generalized lattice market model incorporates both serially correlated returns and asset correlation through a conditional probabilistic model.
	In the nominal case, where the parameters of the model are known, we demonstrate that the proposed policies ensure survivability and probabilistic positivity. 
	We then derive an analytic expression for the worst-case expected gain-loss and prove sufficient conditions that the proposed policies can maintain a \textit{positive expected profits}, even within a seemingly nonprofitable symmetric lattice market. 
	When the parameters are unknown and require estimation, we show that the parameter space of the lattice model forms a convex polyhedron, and we present an efficient estimation method using a constrained least-squares method.
	These theoretical findings are strengthened by extensive empirical studies using data from the top 30 companies within the S\&P 500 index, substantiating the efficacy of the generalized model and the robustness of the proposed policies in sustaining the positive expected profit and providing downside risk protection.
\end{abstract}
%




%


\section{Introduction} 	
The robustness of algorithmic trading systems has been a focal point of numerous studies such as \cite{dokuchaev1998asymptotic, dokuchaev2002dynamic, korajczyk2004momentum, barmish2008trading}.
Among the various aspects of robustness,  the Robust Positive Expectation~(RPE), a trading policy capable of generating a positive expected profit across various market conditions, is particularly sought after.  
Notably, among the strategies targeting RPE, the so-called \textit{Simultaneous Long-Short~(SLS)} strategy, proposed by \cite{barmish2011arbitrage, barmish2015new}, has proven to be a significant advancement.
This strategy involves investors holding both long and short positions simultaneously, with equal weights, that leverages market movements in either direction.

Subsequently, several extensions have been proposed to the SLS strategy, including generalization for Merton’s diffusion model in~\cite{baumann2016stock}, geometric Brownian motion (GBM) model with time-varying parameters in~\cite{primbs2017robustness}, and any linear stochastic differential equation (SDE) in~\cite{baumann2019positive}. 
Additionally, the SLS strategy has also been extended to the proportional-integral (PI) controller in \cite{malekpour2018generalization}, to latency trading in \cite{malekpour2016stock},  and coupled SLS strategy on pair trading for two correlated assets was studied in \cite{deshpande2018generalization, deshpande2020simultaneous}.
Recent contributions to the SLS theory include the work by \cite{baumann2023theoretical}.

Other innovative approaches have focused on robust stock trading. For example,
\cite{maroni2019robust} proposed a robust design strategy for
stock trading via feedback control.
\cite{primbs2018pairs, tie2018optimal} studied the stochastic control-theoretic approach in a pair trading framework, 
\cite{vitale2018robust} investigated robust trading for ambiguity-averse insider,
\cite{o2020generalized} proposed a generalized SLS with different weight settings on long and short positions. 
Recently, \cite{hsieh2022robust} introduced a new variant of the SLS strategy, termed the \textit{double linear policy}, which assigns equal weights to long and short positions in a discrete-time setting, creating a mean-variance criterion to determine optimal weights. This work was later extended by \cite{wang2023robustness} to involve time-varying weights, under the assumption of serially independent returns.

Despite numerous contributions to the robust trading systems, a crucial gap persists in the existing literature. 
Most existing strategies, including SLS or double linear policy, are limited to single or paired assets, and often model asset correlation within the trading strategy, rather than considering it as a characteristic of the market. 
More importantly, they often overlook serial dependence in asset returns, a notable empirical phenomenon \cite{fielitz1971stationarity, officer1972distribution, campbell1993trading, christoffersen2006financial}.
While the work by \cite{balvers1997autocorrelated} investigated autocorrelated returns and optimal intertemporal portfolio choice, it primarily focuses on short-term memory, and more importantly, it does not address the robustness of the strategy in various market conditions.
These omissions pose a significant challenge in robust trading strategies, as emphasized by recent work that highlights the impact of tick-size reductions on various stocks and the influence of competing crossing networks~\cite{werner2023tick}.

In response to these gaps,  this paper introduces an extension to the double linear policy, termed \textit{multi-double linear policies}.
We rigorously investigate its stochastic robustness in a generalized lattice market that involves both serial and asset correlations, as detailed in Section~\ref{section: Preliminaries}.
Then we demonstrate that the proposed policies ensure \textit{survivability} and \textit{probabilistic positivity}.
Following this, we establish a detailed gain-loss analysis and derive an analytic expression of the worst-case expected gain-loss function.
We then prove that the proposed policies offer a form of downside risk protection in comparison to the pure long-only and pure short-only strategies.
We demonstrate an ``approximate" RPE in Section~\ref{section: gain-loss analysis} and establish sufficient conditions under which the positive expected profit is upheld, even within a symmetric lattice market, thereby underscoring the theoretical foundations of our approach.
We further show that the parameters used in the proposed generalized lattice model form a convex polyhedron, which facilitates efficient estimation through a constrained least-squares approach; see Section~\ref{section: parameters estimation}. 
These theoretical findings are substantiated by extensive empirical studies, lending practical relevance to the developed framework; see Section~\ref{section: empirical studies}.

\medskip
\section{Preliminaries and Notations} \label{section: Preliminaries}
This section provides an overview of the formulation, including the generalized lattice market with binomial serially correlated returns and asset correlation, the multi-double linear policies, and the robust positive expectation~property.

\medskip
\subsection{Generalized Lattice Market with Serial and Asset Correlations} \label{subsection: Generalized Lattice Market with Serial and Asset Correlations}
Consider a lattice market consisting of $n \geq 1$ distinct assets. 
For each asset, say Asset~$i$ with~$i \in \{1,2,\dots,n\}$, the \textit{per-period} returns are represented by the sequence $\{X_i(k) : k \geq 0\}$. These returns exhibit two possible outcomes, with $X_i(k) \in \{u_i, d_i\}$, where $u_i \in (0,1)$ represents the \textit{upward movement factor}\footnote{When working with price data on a daily based or shorter timescale,  the values of the movement factors are generally observed to be small, i.e., $|u_i| \ll 1$ and  $|d_i| \ll 1$ for all $i=1,\dots,n$; see \cite{granger2004occasional}.} and~$d_i \in (-1, 0)$ represents the \textit{downward movement factor}.
In combination, this holds for all stages $k$ satisfying 
$
-1 < d_i < 0 < u_i < 1
$
for $i = 1, 2, \dots, n$.
The generalization of the model stems from the assumption that there are asset correlations and serial correlations in the probability of returns, signifying that a positive return at each stage depends on the previous stages and other assets.

In particular, we assume that the return process $\{X_{i}(k): k \geq 0\}$  with $X_{i}(k) \in \{u_i, d_i\}$, exhibits Markov memory with serial autocorrelation and asset correlation. 
Specifically, take the vector of returns at stage $k$ as $X(k) := [X_1(k) \; X_2(k) \cdots X_n(k)]^\top$, the vector of \textit{upward movement factors} as $U:=[u_1\; u_2\; \cdots \; u_n]^\top$, and the vector of \textit{downward movement factors} as~$D:=[d_1\; d_2\; \cdots d_n]^\top$.
Henceforth, we use the shorthand $X(k-1:k-m) := (X(k-1), X(k-2), \dots, X(k-m))$.
Denote the matrix of autocorrelation coefficients as $\Phi:=[\Phi_{i,j}] \in \mathbb{R}^{n \times m+1}$ where~$i = 1,2,\dots,n$ and~$j=0, 1,2,\dots,m$ and the matrix of asset correlation coefficients as $\Gamma:=[\Gamma_{i,\ell}] \in \mathbb{R}^{n \times n}$ where~$i,\ell=1,2,\dots,n$.
Let $\textbf{P}^U(k; X(k-1:k-m))$ be the {vector} of conditional probabilities of positive returns at stage $k$, affected by the previous $m$ stages and asset correlations.
The $i$th component of this vector is given by $[\textbf{P}^U(k; X(k-1:k-m))]_i := \mathbb{P} \left( X_i(k) = u_i \mid X(k-1:k-m) \right) \in [0,1]$. Hence,
\begin{align}
\label{eq: conditional probability}
\textbf{P}^U(k; X(k-1:k-m)) 
&=  \Phi_0 +   \begin{bmatrix}
	\sum_{j=1}^m \Phi_{1,j} X_1(k-j) \\ \vdots\\ \sum_{j=1}^m \Phi_{n,j} X_n(k-j)
\end{bmatrix}  + \Gamma X(k-1).
\end{align}
where $\Phi_0$ is the first column of the matrix $\Phi=[\Phi_0, \Phi_1,\dots, \Phi_m]$.
Therefore, for Asset $i$, the conditional probability 
\begin{align} \label{eq: conditional probability for asset i}
\mathbb{P} \left( X_i(k) = u_i \mid X(k-1:k-m) \right)
&= \Phi_{i,0} + \sum_{j=1}^m \Phi_{i,j} X_i(k-j) + \sum_{\ell=1}^n \Gamma_{i, \ell} X_\ell(k-1)  \in [0,1],
\end{align}
where $\Phi_{i,0}, \Phi_{i,1}, \dots, \Phi_{i,m}$ are the elements of the~$i$th row of the matrix $\Phi$ and $\Gamma_{i, \ell}$ represents the~$i$th row and $\ell$th column of the matrix $\Gamma$
and the conditional probabilities include initial conditions~$ X_{i}(-j) = x_{i}(-j) \in \{u_i, d_i\}$ for $ j= 1, 2, \dots, m $. 
Figure \ref{fig: binomial lattice} illustrates the idea of the generalized lattice model with Markov memory, specifically for $m=1$ (indicating one memory length) and $n=2$ (representing two distinct assets).
The figure includes the initial realized return~$(x_1(-1), x_2(-1) )$ for the two assets, demonstrating the probability for Asset~1 at each stage is dependent on the realized returns from the preceding stage.

Our analysis assumes that the trades incur zero transaction costs and that the underlying asset has perfect liquidity. Fractional shares are allowed.
This setting serves as an appropriate starting point for building the model and aligns closely to the standard \textit{frictionless market} framework in finance; see~\cite{cvitanic2004introduction, luenberger2013investment}. 
While not analyzed in theory, later in Section~\ref{section: empirical studies}, some of the assumptions such as transaction costs are relaxed in empirical studies.

\medskip
\begin{remark} \rm
$(i)$ Returns often exhibit serial correlation, see \cite{fielitz1971stationarity, campbell1993trading, christoffersen2006financial}, with past returns over several periods. This can be due to factors like momentum, mean reversion, or other market dynamics, allowing the model to capture long-term dependencies in the return process. 
$(ii)$  The correlation matrix $\Gamma$ can be extended to time-dependent correlations~$\Gamma(k) = f(X(k-1: k-m))$ where $f$ is a function that captures the dependence of the correlation on past returns.
\end{remark}

\begin{figure}[h!]
\centering
\includegraphics[width=.7\linewidth]{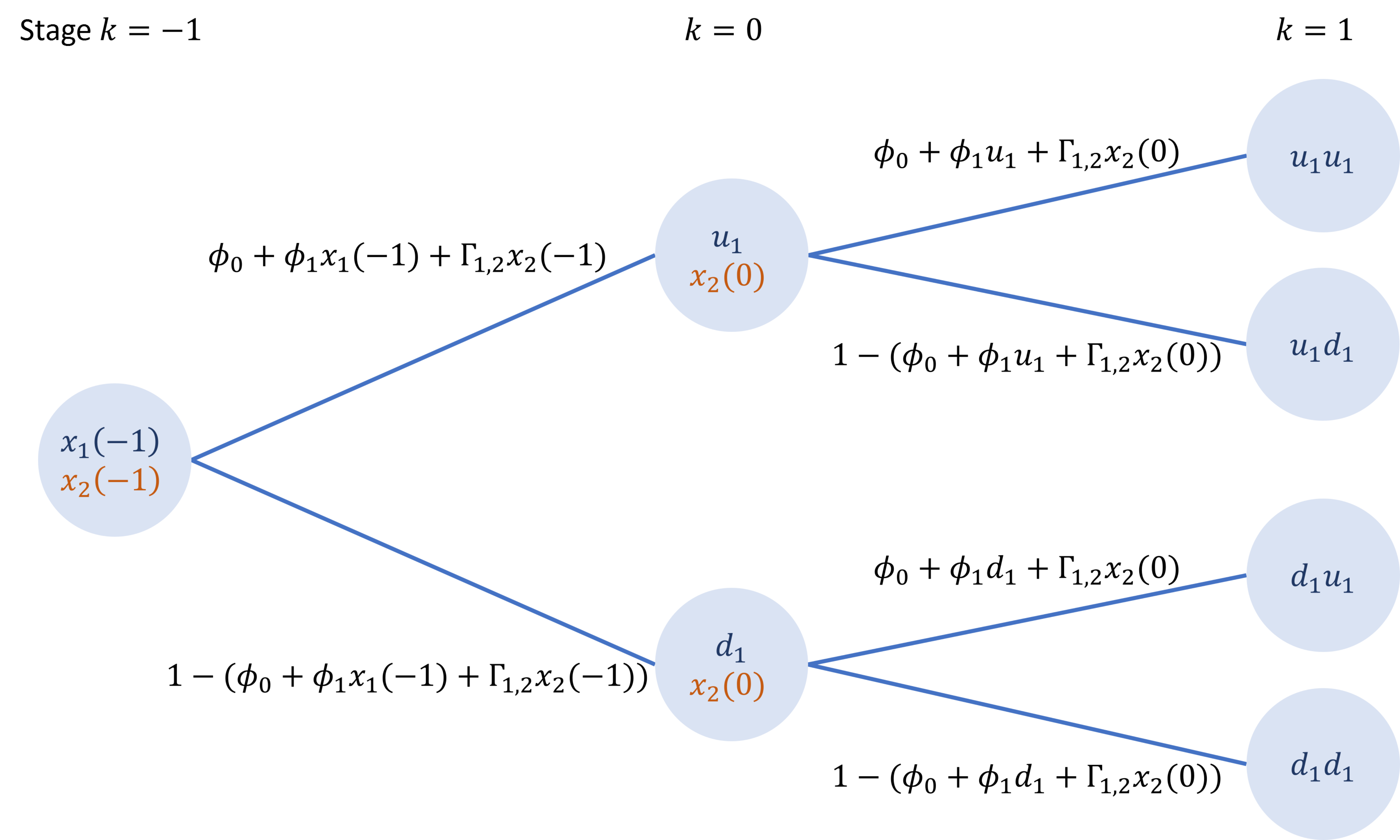}
\caption{An Illustration of Generalized Lattice Model.}
\label{fig: binomial lattice}
\end{figure}

\medskip
\subsection{Multi-Double Linear Policies}
Building upon~\cite{hsieh2022robust, hsieh2022robustness}, we generalize the single double linear policy, initially constructed for the single asset case, to multi-double linear policies for multi-asset portfolio case. 
This generalization allows for trading a portfolio consisting of $n \geq 1$ distinct assets. 
Beginning with an initial account value of an investor,~$V(0) := V_0 > 0$, which is to be allocated across the $n$ assets using the fractions~$v_i \in [0, 1]$ for $i=1,2,\dots,n$ satisfying $\sum_{i=1}^{n} v_i = 1$.
That is, the initial account value for Asset~$i$ is given by~$V_{i,0} := v_i V_0$ and the initial account $V(0)$ is a convex combination of $V_{i,0}$.

Thereafter, we strategically divide it into two parts for each asset, representing~\textit{long} and \textit{short} positions.
Superficially, take a fraction $\alpha \in [0,1]$, such that the initial account value for the long position of Asset $i$ is~$V_{i, L}(0) := \alpha V_{i,0}$, and the value for the short position is $V_{i, S}(0) := (1-\alpha) V_{i,0}$.
The trading policy for Asset~$i$, denoted by $\pi_i(\cdot)$, is constructed as the sum of the policies for long and short positions, i.e.,
$
\pi_i(k) := \pi_{i,L}(k) + \pi_{i,S}(k),
$
where $\pi_{i, L}$ and $\pi_{i,S}$ are in double linear forms:~$   
\pi_{i, L}(k) = w_i V_{i, L}(k)
$
and 
$
\pi_{i, S}(k) = -w_i V_{i,S}(k),
$
where~$w_i \in \mathcal{W}: = [0,1]$ is the weight\footnote{The choice of $\cal W=[0, 1]$ assures that the account is always \textit{survival}, i.e., the investor who adopts the double linear policy with the weight $w_i \in \mathcal{W}$ never goes broke with probability one; see Lemma~\ref{lemma: Survivaliblity}.} for Asset~$i$. 
Here, $V_{i, L}(k)$ and $V_{i, S}(k)$ represent the account value for the long and short positions of Asset~$i$ at stage $k$, respectively. 
The long-short account value dynamics for Asset~$i$ are given by
\begin{align*}
\begin{cases}
	V_{i,L}(k+1) = V_{i,L}(k) + \pi_{i,L}(k) X_i(k) + (V_{i,L}(k)-\pi_{i,L}(k))r_f; \\
	V_{i,S}(k+1) = V_{i,S}(k) + \pi_{i,S}(k)X_i(k),
\end{cases}
\end{align*}
where $r_f \geq 0$ is the risk-free rate.\footnote{
Note that the risk-free rate $r_f$ is included in the dynamics of the long position to signify the opportunity cost of holding the asset; conversely, in the short position, where the asset is borrowed rather than invested, this leads to the absence of $r_f$ in the short position dynamics.}
The overall account value at stage $k$ is given
\begin{align*}
V(k) 
&= \sum_{i=1}^{n} (V_{i,L}(k) + V_{i,S}(k)) \\
& = \sum_{i=1}^{n} V_{i,0}\left(\alpha R_{i+}(k) + (1-\alpha) R_{i-}(k)\right),
\end{align*}
where $V_{i,0} = v_i V_0$, $V_{i,L}(k) =  \alpha V_{i,0}  R_{i+}(k)$, and $V_{i,S}(k) = (1-\alpha)  V_{i,0} R_{i-}(k)$ with $R_{i+}(k) := \prod_{j=0}^{k-1}\left( (1+r_f) + w_i (X_i(j) - r_f)\right)$ and $R_{i-}(k):= \prod_{j=0}^{k-1}\left(1-w_iX_i(j)\right)$.
Note that if~$\alpha=1$, the policies correspond to a pure long-only position, while~$\alpha =0$ corresponds to a pure short position. Therefore, one immediately sees that the proposed policies greatly generalize the  common trading~strategies.

\medskip
\begin{remark}\rm
It should be noted that the notations $V(k)$ or $V_{i, L}(k)$ or $V_{i, S}(k)$, which implicitly depend on various parameters and variables, are suppressed for the sake of simplicity. 
In particular, the overall account value at stage $k$ can be expressed as  $V(k) = V\left( k; ( \alpha, w, v), r_f, \{X(j)\}_{j=0}^{k-1}, V_0 \right)$ where~$(\alpha, w, v) \in [0,1] \times [0, 1]^n \times \mathcal{V}$ is the triple for the proposed polices with $\mathcal{V}:=\{v \in \mathbb{R}^n: v_i \in [0,1], \; \sum_{i=1}^n v_i=1\}$, $r_f \geq 0$ is the risk-free rate, and~$\{X(j)\}_{j=0}^{k-1}$ denote the sequence of asset returns vectors up to stage $k-1$.
\end{remark}

\medskip
\subsection{Robust Positive Expected Profits} \label{section: robust positive expectation problem}
For $k > 0$, let $\mathcal{G}(k):= V(k) - V_0$ be the \textit{cumulative gain-loss function} up to stage $k$ where~$V(k)$ is the account value at stage $k$. 
Note that, for the sake of simplicity, we express $\mathcal{G}(k) = \mathcal{G}\left( k; (\alpha, w, v), r_f, \{X(j)\}_{j=0}^{k-1}, V_0 \right)$.
Then the expected cumulative gain-loss function is denoted by~$\overline{\mathcal{G}}(k) := \mathbb{E}[\mathcal{G}(k)]$.
The primary stochastic robustness to be studied in this paper is the so-called \textit{robust positive expectation}~(RPE) property, see the definition below.

\medskip
\begin{definition}[Robust Positive Expectation] \rm \label{definition: RPE}
For stage $k = 0, 1, \dots$, let $V_0>0$ be the initial account value, and $V(k)$ be the account value at stage $k$. Define the expected cumulative gain-loss function up to stage $k$ as  
$ 
\overline{\mathcal{G}}(k):= \mathbb{E}[V(k)] - V_0.
$  
A trading policy is said to have a \textit{robust positive expectation}~(RPE) property if it ensures that 
$
\overline{\mathcal{G}}(k)>0
$
for all $k>1$, with nonzero trend and finite~variance. 
\end{definition}

\medskip
\subsection{Notation}
We denote by $\mathbb{R}_+$ the nonnegative reals and $\mathbb{R}_{++}$ the strictly positive reals.
For a set $A \in \mathbb{R}$, the indicator $\mathbb{1}_A$ is defined as $\mathbb{1}_A(x)=1$ if $x \in A$; otherwise $\mathbb{1}_A(x)=0$. 
For $a,b\in \mathbb{R}$ and $a<b$ and $n$ a positive integer, a set $(a,b)^n$ or $[a,b]^n$ with  refers to the Cartesian product of the interval $(a,b)$ or~$[a,b]$ with itself $n$-times.

\medskip
\section{Stochastic Robustness of the Policies} \label{section: gain-loss analysis}
In this section, we analyze the stochastic robustness of the multi-double linear policies within the generalized lattice market with both serial and asset correlations.
To facilitate our analysis, we use the vector notations $v:=[v_1 \; v_2\; \cdots \; v_n]^\top \in \mathcal{V}$ where the set $\mathcal{V}:=\{v \in \mathbb{R}^n: v_i \in [0,1], \; \sum_{i=1}^n v_i=1\}$, and $w :=[w_1 \; w_2 \; \cdots \; w_n]^\top \in [0,1]^n$.
With these notations, the multi-double linear policies, as described in Section~\ref{section: Preliminaries}, are characterized by the triple~$(\alpha, w, v) \in [0, 1] \times [0, 1]^n \times \mathcal{V}.$	It is important to note that our analysis in this section assumes  a \textit{nominal} case in the sense that the generalized lattice model's parameters $u_i, d_i, \Phi_{i,j},$ and $ \Gamma_{i,\ell}$ are known. Subsequently, in Section~\ref{section: parameters estimation}, we provide an efficient least-square-based approach to estimate these parameters in a practical~context.

\medskip
\subsection{Survivability and Probabilistic Positivity}
The first result indicates that, for each Asset $i$, the multi-double linear policies assure survival trades; i.e., bankruptcy is avoided.

\medskip
\begin{lemmarep}[Survivability] \label{lemma: Survivaliblity}
For multi-double linear policies with  triple $(\alpha, w, v) \in [0,1] \times [0, 1]^n \times \mathcal{V}$ for $i=1, \dots,n$, we have
$V_{i,L}(k) \geq 0$, $V_{i,S}(k) \geq 0$ for all $i$ and all~$k$ with probability one, and hence the overall account $ V(k) \geq 0$ for all $k$ with probability~one.
Moreover, if $v_i>0$ for some~$i$, $w_i \in (0,1)$, and $\alpha \in (0,1)$, then~$V_{i, L}(k)>0$ and $V_{i, S}(k) > 0$ and $V(k) > 0$ for all $k$ with probability~one.
\end{lemmarep}

\medskip
\textit{Proof.} See Appendix~\ref{Appendix: Technical Proofs of Section: gain-loss analysis}.


\medskip
With the aid of Lemma~\ref{lemma: Survivaliblity}, we can prove a general result akin to the Paley-Zygmund inequality, regarding probabilistic positivity.

\medskip
\begin{lemma}[Probabilistic Positivity] \label{lemma: probabilistic positivity}
Fix $k > 1$. For multi-double linear policies with triple $(\alpha, w, v) \in (0,1) \times (0,1)^n \times \mathcal{V}$, it follows that
\begin{align} \label{ineq: probabilistic positivity}
	\mathbb{P}(V(k) > \theta \mathbb{E}[V(k)] ) \geq  \frac{  (1-\theta)^2 \mathbb{E}[V(k)]^2}{ {\rm var}(V(k)) + (1-\theta)^2 \mathbb{E}[V(k)]^2},
\end{align}
for all $\theta \in [0,1]$.
\end{lemma}

\medskip
\textit{Proof.} See Appendix~\ref{Appendix: Technical Proofs of Section: gain-loss analysis}.


\subsection{Gain-Loss Analysis}
For $k=0, 1, \dots$, let $V(k)$ be the account value at stage $k$. Then the corresponding cumulative gain-loss function, as defined in Section~\ref{section: Preliminaries}, is given by 
$
\mathcal{G}(k) 
= \sum_{i=1}^{n} V_{i,0}\left(\alpha R_{i+}(k) + (1-\alpha) R_{i-}(k) - 1\right), 
$
and the expected cumulative gain-loss function as $\overline{\mathcal{G}}(k) := \mathbb{E}[\mathcal{G}(k)]$.
Note that with this notation, we can rewrite Inequality~(\ref{ineq: probabilistic positivity}) in terms of the cumulative gain-loss function. Indeed, observe that
\begin{align*}
\mathbb{P}( V(k) > \theta \mathbb{E}[V(k)]) 
&= \mathbb{P}( \mathcal{G}(k) > \theta \mathbb{E}[V(k)] - V_0) \\
&=  \mathbb{P}( \mathcal{G}(k)  > \theta \mathbb{E}[V(k)] -\theta V_0  + \theta V_0 - V_0) \\
& =\mathbb{P}( \mathcal{G}(k)  > \theta \overline{\mathcal{G}}(k)  - (1-\theta)  V_0).
\end{align*}
Moreover, noting that $\mathbb{E}[V(k)] = \overline{\mathcal{G}}(k) + V_0$ and ${\rm var}(V(k)) = {\rm var}(\mathcal{G}(k))$, Lemma~\ref{lemma: probabilistic positivity} becomes
\[
\mathbb{P}(\mathcal{G}(k) > \theta \overline{\mathcal{G}}(k)  - (1-\theta) V_0 ) \geq \frac{  (1-\theta)^2 (\overline{\mathcal{G}}(k) + V_0)^2}{ {\rm var}(\mathcal{G}(k)) + (1-\theta)^2 (\overline{\mathcal{G}}(k) + V_0)^2},
\]
which characterizes the probability of cumulative gain-loss larger than the weighted sum of~$\overline{\mathcal{G}}(\cdot)$ and~$V_0$. 
Note that the variance of the cumulative gain-loss function, ${\rm var}(\mathcal{G}(k))$, is given by 
$ {\rm var} (\mathcal{G}(k))
= {\rm var}\left( \sum_{i=1}^n V_{i,0} \left(\alpha R_{i+}(k) + (1-\alpha) R_{i-}(k) -1\right)\right).
$
Although the analytic expression for this variance may be intractable within the generalized lattice market framework, it can be computed numerically; e.g., see later in Section~\ref{section: empirical studies}. 
The next lemma indicates gain-loss under zero weight.

\medskip	
\begin{lemmarep}[Zero Weight Gain-Loss] \label{lemma: zero weight gain-loss}
Consider the multi-double linear policies with triple~$(\alpha, w, v) \in [0,1]\times [0,1]^n \times \mathcal{V}$.
If $w_i = 0$ for all $i=1, \dots, n$, we have 
$
{\mathcal{G}}(k) = \alpha( (1+r_f)^k - 1) V_0 \geq 0
$
for all $k$ with probability one. Moreover, if the risk-free rate is zero, i.e., $r_f =0$, then the gain-loss function simplifies to $\mathcal{G}(k) = 0$ for all~$k$.
\end{lemmarep}

\medskip
\textit{Proof.} See Appendix~\ref{Appendix: Technical Proofs of Section: gain-loss analysis}.


\medskip
To prove our first main result, Theorem~\ref{thm: proof of RPE}, the following two auxiliary lemmas are useful. The first one states that the probability of receiving a positive return can be characterized by a recursion formula, and the second one expresses a sum of expected logarithmic functions using the probability recursion.

\medskip
\begin{lemmarep}[A Recursion Formula of Probability] \label{lemma: recursion formula of probability}
For $i = 1,2, \dots,n$, assume that $X_i(\cdot) \in \{u_i, d_i\}$ are returns characterized by the generalized lattice model as described in Section~\ref{subsection: Generalized Lattice Market with Serial and Asset Correlations}. For $k \geq 1$, let $p_i(k):= \mathbb{P} (X_i(k) = u_i)$ be the probability of receiving a positive return of Asset $i$ at stage~$k$ with initial conditions $
p_{i}(-j) = \frac{x_{i}(-j) - d_i}{u_i - d_i}
$ for $j=1, 2, \dots, m$ with $x_i(-j) \in \{u_i, d_i\}$. 
Then  $p_i(k)$ satisfies the following recursion
\begin{align} \label{eq: p_k recursion}
p_i (k) 
&= \Phi_{i,0} + \sum_{j=1}^m \Phi_{i,j} [(u_i - d_i) \cdot p_i({k-j}) + d_i]+ \sum_{\ell=1}^n \Gamma_{i, \ell} [(u_\ell - d_\ell) \cdot p_\ell({k-1}) + d_\ell].
\end{align}
\end{lemmarep}

\medskip
\textit{Proof.} See Appendix~\ref{Appendix: Technical Proofs of Section: gain-loss analysis}.


\medskip
\begin{remark} \rm
Note that the returns $X_i(\cdot) \in \{u_i, d_i\}$ characterized by the generalized lattice model, follows a Bernoulli distribution with time-varying probability $p_i(\cdot)$, i.e., $X_i(\cdot) \sim Bernoulli(p_i(\cdot))$.
Later in Section~\ref{section: empirical studies},  the recursion formula~(\ref{eq: p_k recursion}) is useful to estimate the binomial probability for each stage. 
\end{remark}

\medskip
\begin{lemmarep} \label{lemma: the expectation of Binomial distribution}
Let $k \geq 1$ and consider assets $i=1, \dots,n$, each associated with a weight $w_i \in [0, 1]$. 
Assume that $X_i(\cdot) \in \{u_i, d_i\}$ are returns characterized by a generalized lattice model as described in Section~\ref{subsection: Generalized Lattice Market with Serial and Asset Correlations}. Then the following equality holds:\footnote{The notation $\pm$ used here represents the positive and negative variations in Equation~(\ref{eq: sum of expected logarithmic function}). If $\pm$ on the left-hand side of Equation~(\ref{eq: sum of expected logarithmic function}) is $+$, then all instances of $\pm$ on the right-hand sides are $+$. Likewise, if the left-hand side is $-$, then all instances of $\pm$ on the right-hand side becomes $-$.}
\begin{align} \label{eq: sum of expected logarithmic function}
\sum_{j=0}^{k-1} \mathbb{E}[\log (1\pm w_i X_i(j))] 
&= \mathbb{E}[H_i(k)] \log (1 \pm w_i u_i )  + (k - \mathbb{E}[H_i(k)]) \log (1 \pm w_i d_i),
\end{align}
where $\mathbb{E}[H_i(k)] := \sum_{j=0}^{k-1} p_{i}(j) \leq k$ is the expected number of receiving positive returns up to stage $k-1$, and $p_i(j)$ is stated in Lemma~\ref{lemma: recursion formula of probability}.
\end{lemmarep}

\medskip
\textit{Proof.} See Appendix~\ref{Appendix: Technical Proofs of Section: gain-loss analysis}.


\medskip
\begin{theorem}[The Worst Expected Gain-Loss] \label{thm: proof of RPE}
For~$k>1$, consider the multi-double linear policies with triple $(\alpha, w, v) \in (0, 1) \times (0, 1)^n \times \mathcal{V}$, and assume that the return follows the lattice model described in Lemma~\ref{lemma: the expectation of Binomial distribution}. 
Then the expected gain-loss function $\overline{\mathcal{G}}(k)$ satisfies the following inequality:
\begin{align*}
\overline{\mathcal{G}}(k) 
& > \sum_{i = 1}^n V_{i, 0} \left( \alpha  (\beta_i(k) -1) + (1-\alpha) (\gamma_i(k) -1) \right),
\end{align*}
where  
$\beta_i(k):= ( (1 + r_f) + w_i (u_i - r_f))^{\mathbb{E}[H_i(k) ] }  ( (1+r_f) + w_i (d_i - r_f))^{\left(k-\mathbb{E}[H_i(k)]\right) }  > 0 $
and
$ \gamma_i(k) := (1 - w_i u_i)^{\mathbb{E} [H_i (k)]} (1 - w_i d_i )^{\left(k-\mathbb{E}[H_i(k)]\right) } > 0 $
and~$\mathbb{E}[H_i (k) ] = \sum_{j=0}^{k-1} p_{i}(j)$ with $p_i(j)$ as stated in Lemma~\ref{lemma: recursion formula of probability}.
\end{theorem}

\medskip	
\begin{proof}  
For $i=1,\dots,n$, fix $w_i \in (0, 1).$
The expected cumulative gain-loss function is given by
\begin{align*}
\overline{\mathcal{G}}(k) 
&=  \mathbb{E}\left[   \sum_{i=1}^{n} V_{i,0}\left(\alpha R_{i+}(k) + (1-\alpha) R_{i-}(k) - 1\right) \right] \\
&=  \sum_{i=1}^{n} V_{i,0}  \mathbb{E}\left[    \left(\alpha R_{i+}(k) + (1-\alpha) R_{i-}(k) - 1\right) \right].
\end{align*}
Note that the expectation term can be written as
\begin{align*}
&\mathbb{E}\left[ \alpha R_{i+}(k) + (1-\alpha) R_{i-}(k) - 1 \right] \\
&\quad =  \alpha \mathbb{E} \left[ \prod_{j=0}^{k-1}\left( (1+r_f) + w_i (X_i(j) - r_f)\right) \right] + (1-\alpha) \mathbb{E} \left[ \prod_{j=0}^{k-1} \left(1 - w_i X_i(j)\right)\right] -1 \\[2ex]
&\quad = \alpha \mathbb{E} \left[ e^{\log  \prod_{j=0}^{k-1}\left( (1+r_f) + w_i (X_i(j) - r_f)\right) } \right] + (1-\alpha) \mathbb{E} \left[ e^{\log \prod_{j=0}^{k-1} \left(1-w_i X_i(j)\right) }\right] -1  \\[2ex]
&\quad > \alpha e^{\sum_{j=0}^{k-1} \mathbb{E} \left[ \log \left( (1+r_f) + w_i (X_i(j) - r_f ) \right) \right] } + (1-\alpha) e^{\sum_{j=0}^{k-1} \mathbb{E} \left[ \log \left(1-w_iX_i(j)\right) \right]} -1, 
\end{align*}
where the last inequality holds by Jensen's inequality on strictly convex function $\exp(\cdot)$, i.e., $\mathbb{E}[e^Z] > e^{\mathbb{E}[Z]}$ for some random variable $Z$. 
Subsequently, applying Lemma~\ref{lemma: the expectation of Binomial distribution}, we have
\begin{align*}         
\overline{\mathcal{G}}(k) 
& >  \sum_{i = 1}^n V_{i, 0} \Big( \alpha e^{ \mathbb{E}[H_i(k) ] \log ( (1 + r_f) + w_i (u_i - r_f)) + \left(k-\mathbb{E}[H_i(k)]\right) \log ( (1+r_f) + w_i (d_i - r_f)) }  \nonumber \\
&\qquad \qquad \quad + (1-\alpha) e^{\mathbb{E} [H_i (k)] \log(1 - w_i u_i)  + \left(k-\mathbb{E}[H_i(k)]\right) \log(1 - w_i d_i ) } -1 \Big)\\
&=  \sum_{i = 1}^n V_{i, 0} \Big( \alpha  ( (1 + r_f) + w_i (u_i - r_f))^{\mathbb{E}[H_i(k) ] }  ( (1+r_f) + w_i (d_i - r_f))^{\left(k-\mathbb{E}[H_i(k)]\right) }  \nonumber \\
&\qquad \qquad \quad + (1-\alpha) (1 - w_i u_i)^{\mathbb{E} [H_i (k)]} (1 - w_i d_i )^{\left(k-\mathbb{E}[H_i(k)]\right) } -1 \Big)\\
&= \sum_{i = 1}^n V_{i, 0} \left( \alpha  (\beta_i(k) -1) + (1-\alpha) (\gamma_i(k) -1) \right),
\end{align*}
where $
\beta_i(k):= ( (1 + r_f) + w_i (u_i - r_f))^{\mathbb{E}[H_i(k) ] }  ( (1+r_f) + w_i (d_i - r_f))^{\left(k-\mathbb{E}[H_i(k)]\right) }  > 0
$
and
$
\gamma_i(k) := (1 - w_i u_i)^{\mathbb{E} [H_i (k)]} (1 - w_i d_i )^{\left(k-\mathbb{E}[H_i(k)]\right) } > 0.
$
The positivity holds since $r_f \geq 0$, $-1< d_i <0 < u_i < 1$, and $w_i \in (0,1)$ for all $i$.
\end{proof}

\medskip
\begin{corollaryrep}[Some Special Cases] \label{corollary: some special cases}
Consider the multi-double linear policies with triple~$(\alpha, w, v) \in \{1/2\} \times (0,1)^n \times \mathcal{V}$.
We have the following conclusions:
\begin{enumerate}
\item[(i)] If $\mathbb{E}[H_i(k)] \in \{0, k\}$, then
$\overline{\mathcal{G}}(k) > 0
$ for all $k > 1$.
\item[(ii)] If $\mathbb{E}[H_i(k)] \in (0, k)$ and $(1+w_iu_i)^{\mathbb{E}[H_i(k)]} (1+w_i d_i)^{k - \mathbb{E}[H_i(k)]} + (1- w_i u_i)^{\mathbb{E}[H_i(k)]} (1 - w_i d_i)^{k - \mathbb{E}[H_i(k)]} > 2$ for all $i =1,\dots, n$, then~$\overline{\mathcal{G}}(k) > 0$ for all $k > 1$.
\end{enumerate}	
\end{corollaryrep}

\medskip
\textit{Proof.} See Appendix~\ref{Appendix: Technical Proofs of Section: gain-loss analysis}.


\medskip
\begin{remark}[Approximate RPE] \rm
\label{remark: approximation error}
$(i)$ Consider the special case where the risk-free rate $r_f =0$, and let both $w_iu_i \approx 0$ and $w_i d_i \approx 0$. 
Then, for each Asset~$i$ and fixed~$k$, we have~$\beta_i(k) \approx 1$ and $\gamma_i(k) \approx 1$, which implies that the lower bound in Theorem \ref{thm: proof of RPE} is given by~$ \sum_{i = 1}^n V_{i, 0} \left( \alpha  (\beta_i(k) -1) + (1-\alpha) (\gamma_i(k) -1) \right) \approx 0$ for $k$. 
This leads to the inequality $\mathcal{G}(k) \gtrapprox 0$ for all~$k$, provided that $w_iu_i$ and $w_id_i$ are sufficiently small.
Said another way, we observe an ``approximate" RPE property.
$(ii)$ As seen later in the next section, the multi-double linear policies proposed in this paper can result in a positive expected profit. This holds true even in markets with a seemingly unprofitable symmetric lattice market. See also Section~\ref{section: empirical studies} for analysis of positive expected profits in practice.
\end{remark}

\medskip
\subsection{Positive Expected Profits When Market has Clear Trends}
The main result of this section indicates that if the assets in the market have clear trends; i.e.,~$u_i > -d_i$ and $\mathbb{E}[H_i(k)] > k/2$ (or $u_i < -d_i$ and $\mathbb{E}[H_i(k)] < k/2$), then it is possible to establish sufficient conditions for the multi-double linear policies to ensure the positive expected profits. To prove such a result, the following lemma is useful.

\medskip
\begin{lemma}[A Strict Convex Auxiliary Function]\label{lemma: strictly convexity}
Consider a function $\theta: \mathbb{R}_{++} \to \mathbb{R}$ with $\theta(\varepsilon):= a^\varepsilon + b^\varepsilon$ where $a,b>0$ with $a,b \neq 1$. Then $\theta(\varepsilon)$ is strictly convex in $\varepsilon$.
Moreover, if either $a>1$ and~$b<1$ or $a<1$ and $b>1$, then $\phi(\varepsilon)$ has a unique minimum at some $\varepsilon= \varepsilon^*>0$ satisfying~$a^{\varepsilon^* } \log a + b^{\varepsilon^* } \log b =0$.
\end{lemma}	

\medskip
\textit{Proof.} See Appendix~\ref{Appendix: Technical Proofs of Section: gain-loss analysis}.


\medskip
\begin{theorem}[Positive Expected Profits]
Consider multi-double linear policies with triple~$(\alpha, w, v) \in \{1/2\} \times (0,1)^n \times \mathcal{V}$ within the generalized lattice market. Then the following two statements hold true.

$(i)$ Suppose $u_i > -d_i $ for all $i$ and $\mathbb{E}[H_i(k)] = k/2 +\varepsilon_i$ for~$\varepsilon_i >  \theta^{-1} \left( \frac{2}{ (1 -  w_i^2 d_i^2 - w_i \delta_i + w_i^2\delta_i d_i )^{k/2} } \right) > \varepsilon_i^*$ with 
$\theta(\varepsilon_i) = \left( \frac{  1-  w_i d_i + w_i \delta_i  }{1+ w_i d_i} \right)^{ \varepsilon_i}  +   \left(   \frac{  1 +  w_i d_i - w_i \delta_i   }{1 - w_i d_i} \right)^{ \varepsilon_i}$ and $ \varepsilon_i^*$ satisfying
\[
\left(\frac{1 -  w_i d_i + \delta_i  w_i }{1 + w_i d_i  }\right)^{\varepsilon_i^* }  \log \left(\frac{1 -  w_i d_i + \delta_i  w_i}{1 +  w_i d_i }\right) +   \left(\frac{ 1 +  w_i d_i - \delta_i  w_i}{1 - w_i d_i }\right)^{\varepsilon_i^* } \log  \left( \frac{1 + w_i d_i  - \delta_i  w_i}{1 - w_i d_i }\right) = 0,
\]
then $\overline{\mathcal{G}}(k) > 0$.

$(ii)$	Suppose $u_i < -d_i $ for all $i$ and $\mathbb{E}[H_i(k)] = k/2  - \varepsilon_i$ for~$\varepsilon_i > \vartheta^{-1} \left( \frac{2}{ (1-  w_i^2 d_i^2 - w_i \delta_i - w_i^2\delta_i d_i )^{k/2} } \right)  > \varepsilon_i^* $ with $\vartheta(\varepsilon_i):= \left( \frac{ 1+ w_i d_i  }{ 1-  w_i d_i - w_i \delta_i } \right)^{ \varepsilon_i}  +   \left( \frac{ 1 - w_i d_i  }{ 1 +  w_i d_i +  w_i \delta_i  }\right)^{ \varepsilon_i}$  and $ \varepsilon_i^*$ satisfying
\[
\left( \frac{ 1+ w_i d_i  }{ 1-  w_i d_i - w_i \delta_i } \right)^{\varepsilon_i^* }  \log \left( \frac{ 1+ w_i d_i  }{ 1-  w_i d_i - w_i \delta_i } \right) +   \left( \frac{ 1- w_i d_i  }{ 1+  w_i d_i + w_i \delta_i } \right)^{\varepsilon_i^* } \log  \left( \frac{ 1- w_i d_i  }{ 1+  w_i d_i + w_i \delta_i } \right) = 0,
\]
then $\overline{\mathcal{G}}(k) >0$.
\end{theorem}

\begin{proof}
We shall only give proof for part~$(i)$ since an almost identical argument would work for part~$(ii)$.
To prove part~$(i)$, let $u_i > -d_i$. Then  there exists $0 < \delta_i < 1+ d_i$ such that $u_i = -d_i + \delta_i$.
With the triple~$(\alpha, w, v) \in \{1/2\} \times (0,1)^n \times \mathcal{V}$, Theorem~\ref{thm: proof of RPE} indicates that the expected gain-loss function satisfies
$
\overline{\mathcal{G}}(k) > \sum_{i = 1}^n \frac{V_{i, 0}}{2} \left(  \beta_i(k) +  \gamma_i(k) - 2 \right),
$
where
\begin{align*}
\beta_i(k) 
&= (1+ w_i (-d_i + \delta_i))^{k/2 + \varepsilon_i}(1+ w_i d_i)^{k/2 - \varepsilon_i}\\
&= (1-  w_i d_i + w_i \delta_i )^{k/2}(1-  w_i d_i + w_i \delta_i )^{ \varepsilon_i}(1+ w_i d_i)^{k/2 } (1+ w_i d_i)^{ - \varepsilon_i}\\
&= (1-  w_i^2 d_i^2 + w_i \delta_i + w_i^2\delta_i d_i )^{k/2} \left( \frac{  1-  w_i d_i + w_i \delta_i  }{1+ w_i d_i} \right)^{ \varepsilon_i} ,
\end{align*}
and 
\begin{align*}
\gamma_i(k) 
&= (1 - w_i (-d_i + \delta_i))^{k/2 + \varepsilon_i}(1- w_i d_i)^{k/2 - \varepsilon_i}\\
&= (1 +  w_i d_i - w_i \delta_i )^{k/2}(1 +  w_i d_i - w_i \delta_i )^{ \varepsilon_i}(1 - w_i d_i)^{k/2 } (1 -  w_i d_i)^{ - \varepsilon_i}\\
&= (1-  w_i^2 d_i^2 - w_i \delta_i + w_i^2\delta_i d_i )^{k/2} \left(   \frac{  1 +  w_i d_i - w_i \delta_i   }{1 - w_i d_i} \right)^{ \varepsilon_i}.
\end{align*}	
Therefore, we have
\begin{align*}
\overline{\mathcal{G}}(k) 
&> \sum_{i = 1}^n \frac{V_{i, 0}}{2} \bigg(  (1-  w_i^2 d_i^2 + w_i \delta_i + w_i^2\delta_i d_i )^{k/2} \left( \frac{  1-  w_i d_i + w_i \delta_i  }{1+ w_i d_i} \right)^{ \varepsilon_i}  \\
&\qquad \qquad \qquad +  (1-  w_i^2 d_i^2 - w_i \delta_i + w_i^2\delta_i d_i )^{k/2} \left(   \frac{  1 +  w_i d_i - w_i \delta_i   }{1 - w_i d_i} \right)^{ \varepsilon_i} - 2 \bigg)\\
& >   \sum_{i = 1}^n \frac{V_{i, 0}}{2} \bigg(  (1-  w_i^2 d_i^2 - w_i \delta_i + w_i^2\delta_i d_i )^{k/2} \left[ \left( \frac{  1-  w_i d_i + w_i \delta_i  }{1+ w_i d_i} \right)^{ \varepsilon_i}  +   \left(   \frac{  1 +  w_i d_i - w_i \delta_i   }{1 - w_i d_i} \right)^{ \varepsilon_i}\right]  - 2 \bigg)\\
&= \sum_{i = 1}^n \frac{V_{i, 0}}{2} \bigg(  (1-  w_i^2 d_i^2 - w_i \delta_i + w_i^2\delta_i d_i )^{k/2} \cdot \theta(\varepsilon_i)  - 2 \bigg),
\end{align*}
where $\theta(\varepsilon_i):=a_i^{ \varepsilon_i}  +   b_i^{ \varepsilon_i}$ is an auxiliary function with $a_i :=  \frac{  1-  w_i d_i + w_i \delta_i  }{1+ w_i d_i} $ and $b_i :=   \frac{  1 +  w_i d_i - w_i \delta_i   }{1 - w_i d_i} $. 
Since $w_i \in (0,1)$ and $d_i \in (-1,0)$, it follows that $a_i, b_i >0$ and $a_i > 1$ and $b_i \in (0, 1)$. Therefore, by Lemma~\ref{lemma: strictly convexity}, the auxiliary function $\theta(\varepsilon_i)$ is strictly convex in $\varepsilon_i$ and has  unique minimum at $\varepsilon_i = \varepsilon_i^* >0$ satisfying
\[
\left(\frac{1 -  w_i d_i+   w_i \delta_i}{1 + w_i d_i }\right)^{\varepsilon_i^* }  \log \left(\frac{1 - w_i d_i +  w_i \delta_i}{1 + w_i d_i}\right) +   \left(\frac{ 1 + w_i d_i-  w_i \delta_i }{1 - w_i d_i}\right)^{\varepsilon_i^* } \log  \left( \frac{1 + w_i  d_i-  w_i \delta_i }{1 - w_i d_i }\right) = 0.
\] 
Moreover, the strict convexity implies that for any $\varepsilon_i > \varepsilon_i^*$, the function is strictly increasing; therefore, inverse $\theta^{-1}$ exists and is still strictly increasing for $\varepsilon_i > \varepsilon_i^*$. Using the assumed hypothesis, 
\[
\varepsilon_i > \theta^{-1} \left( \frac{2}{ (1-  w_i^2 d_i^2 - w_i \delta_i + w_i^2\delta_i d_i )^{k/2} } \right),
\]
it follows that $\overline{\mathcal{G}}(k) >0.$
To complete the proof, we must show that $\theta^{-1} \left( \frac{2}{ (1-  w_i^2 d_i^2 - w_i \delta_i + w_i^2\delta_i d_i )^{k/2} } \right)  > \varepsilon_i^*$. Indeed, proceeds a proof by contradiction by assuming that  $\theta^{-1} \left( \frac{2}{ (1-  w_i^2 d_i^2 - w_i \delta_i + w_i^2\delta_i d_i )^{k/2} } \right)  \leq \varepsilon_i^*$. 
Within this range, $\theta$ is strictly decreasing, therefore, it follows that
\begin{align} \label{ineq: theta condition}
\frac{2}{ (1-  w_i^2 d_i^2 - w_i \delta_i + w_i^2\delta_i d_i )^{k/2} }  \geq  \theta( \varepsilon_i^*).
\end{align}
However, note that
$
\theta(\varepsilon_i^*)  = \inf_{\varepsilon_i} \theta(\varepsilon_i) \leq \theta(0) = 2.
$
Hence, it follows that 
\[
(1-  w_i^2 d_i^2 - w_i \delta_i + w_i^2\delta_i d_i )^{k/2}  \theta( \varepsilon_i^*) - 2 \leq  2( (1-  w_i^2 d_i^2 - w_i \delta_i + w_i^2\delta_i d_i )^{k/2} - 1) < 0.
\]
This implies that
$
\theta( \varepsilon_i^*)  < \frac{2}{(1-  w_i^2 d_i^2 - w_i \delta_i + w_i^2\delta_i d_i )^{k/2}  }
$
which contradicts to Inequality~(\ref{ineq: theta condition}). 
Therefore, we must have~$\theta^{-1} \left( \frac{2}{ (1-  w_i^2 d_i^2 - w_i \delta_i + w_i^2\delta_i d_i )^{k/2} } \right)  > \varepsilon_i^*$ and the proof is complete.
\end{proof}

\medskip
\subsection{Gain-Loss Analysis in Symmetric Lattice Market}
An important subclass of the lattice market framework is the \textit{symmetric} market. In this type of market, the factors governing upward and downward movement have the same magnitude, creating a unique structure. 
It should be noted that since the multi-double linear policies earn positive profits when either an upward or downward trend is clear, it is arguable that the robustness of the symmetric lattice market needs to be further scrutinized.
Below, we first provide a corollary for the lower bound of the expected gain-loss of the proposed policies in the symmetric market. Then we prove the  conditions under which the RPE may hold for such a symmetric lattice market.

\medskip	
\begin{corollaryrep}[Limitation of the Polices in Symmetric Lattice Market] \label{corollary: limitations of the policies}
For $k > 1$,
consider the multi-double linear policies with triple $(\alpha, w, v) \in \{1/2\}\times (0,1)^n \times \mathcal{V}$ in a lattice market with $r_f=0$ and  $u_i = -d_i$ for all $i \in \{1,2\dots,n\}$. 
\begin{enumerate}
\item[(i)] If $\mathbb{E}[H_i(k)] \neq k/2$ for all $i \in \{1,2\dots,n\}$, then 
\begin{align*}
	\overline{\mathcal{G}}(k) 
	& > \sum_{i = 1}^n \frac{V_{i, 0} }{2}\left(   (1 - w_i^2 d_i^2)^{k/2 } \left[ \left( \frac{1 - w_i d_i}{1+ w_i d_i} \right)^{\varepsilon_i}   +  \left( \frac{1 - w_i d_i}{1+ w_i d_i} \right)^{-\varepsilon_i} \right] - 2 \right),
\end{align*}
for some $\varepsilon_i \in (0, k/2).$
\item[(ii)] If $\mathbb{E}[H_i(k)] = k/2$  for all $i \in \{1,2\dots,n\}$, then 
$	
\overline{\mathcal{G}}(k) 
> \sum_{i = 1}^n {V_{i, 0} }\left(   [1 - w_i^2 d_i^2]^{k/2 }  - 1 \right).
$
\end{enumerate}
\end{corollaryrep}

\medskip
\textit{Proof.} See Appendix~\ref{Appendix: Technical Proofs of Section: gain-loss analysis}.

\medskip
\begin{remark} \rm
Note that in Part~$(ii)$, if the trade happens within an ``exact" symmetric lattice market, i.e., $u_i = -d_i$ and $\mathbb{E}[H_i(k)] = k/2$, then the lower bound 
$
\sum_{i = 1}^n {V_{i, 0} }\left(   [1 - w_i^2 d_i^2]^{k/2 }  - 1 \right) <0,
$ 
which indicates that RPE may not hold.
However, if $\mathbb{E}[H_i(k)] \neq k/2$, then, with the aid of the following Lemma~\ref{lemma: An Strictly Increasing Auxiliary Function}, it is possible to establish the positive expected profits, see Theorem~\ref{theorem: RPE in Symmetric Lattice Market}.
\end{remark}

\medskip
\begin{lemmarep}[An Strictly Increasing Auxiliary Function] \label{lemma: An Strictly Increasing Auxiliary Function}
Fix $z > 0$ with $z \neq 1$.
Define a function $\phi: \mathbb{R}_{++} \to \mathbb{R}$ with $\phi(\varepsilon):= z^\varepsilon + z^{-\varepsilon}$. Then $\phi(\varepsilon)$ is a strictly increasing and $\inf \phi(\varepsilon) = 2$.  
\end{lemmarep}

\medskip
\textit{Proof.} 
See Appendix~\ref{Appendix: Technical Proofs of Section: gain-loss analysis}.


\medskip

\begin{theorem}[RPE in Symmetric Lattice Market] \label{theorem: RPE in Symmetric Lattice Market}
For $k > 1$, consider the multi-double linear policies with triple $(\alpha, w, v) \in \{1/2\}\times (0,1)^n \times \mathcal{V}$ in a generalized lattice market with $r_f=0$ and~$u_i = -d_i$ for all $i = 1, 2\dots, n$. 
If either $\mathbb{E}[H_i(k)] = k/2 + \varepsilon_i$ or $\mathbb{E}[H_i(k)] = k/2 - \varepsilon_i$ for all $i =1, 2, \dots, n$ satisfying 
$
\varepsilon_i \in \left( \phi^{-1}\left( \frac{2}{(1- w_i^2 d_i^2)^{k/2}} \right), \;  \frac{k}{2} \right)
$
where 
$\phi(\varepsilon_i) := \left( \frac{1 - w_i d_i}{1+ w_i d_i} \right)^{\varepsilon_i}   +  \left( \frac{1 - w_i d_i}{1+ w_i d_i} \right)^{-\varepsilon_i}$,
then
$
\overline{\mathcal{G}}(k) > 0.
$ 
\end{theorem}

\medskip
\begin{proof}
Suppose either $\mathbb{E}[H_i(k)] = k/2 + \varepsilon_i$ or $\mathbb{E}[H_i(k)] = k/2 - \varepsilon_i$ for all $i =1, 2, \dots, n$ and for some~$\varepsilon_i \in (0, k/2)$.
Recalling part~$(i)$ of Corollary~\ref{corollary: limitations of the policies}, we have	
\begin{align} \label{ineq: G_bar in symmetric lattice market}
	\overline{\mathcal{G}}(k) 
	> \sum_{i = 1}^n \frac{V_{i, 0} }{2}\left(   (1 - w_i^2 d_i^2)^{k/2 } \cdot \phi(\varepsilon_i)  - 2 \right),
\end{align}
where $ \phi(\varepsilon_i):= \left( \frac{1 - w_i d_i}{1+ w_i d_i} \right)^{\varepsilon_i}   +  \left( \frac{1 - w_i d_i}{1+ w_i d_i} \right)^{-\varepsilon_i} $ .
Note that the ratio $\frac{1 - w_i d_i}{1+ w_i d_i} \neq 1$ and~$\varepsilon_i > 0$. 
Lemma~\ref{lemma: An Strictly Increasing Auxiliary Function} implies that $\phi(\varepsilon_i)$ is strictly increasing; hence, the inverse~$\phi^{-1}$ exists. 
By the hypothesis, we have $
\varepsilon_i \in \left( \phi^{-1}\left( \frac{2}{(1- w_i^2 d_i^2)^{k/2}} \right), \;  \frac{k}{2} \right)
$
which implies $\phi(\varepsilon_i) > \frac{2}{(1- w_i^2d_i^2)^{k/2}}$. 
Applying this inequality on the right-hand side of Inequality~(\ref{ineq: G_bar in symmetric lattice market}), a simple algebraic manipulation leads to the desired $\overline{\mathcal{G}}(k) >~0$ immediately.
To complete the proof, we must show that
$
\phi^{-1} \left( \frac{2}{(1-w_i^2 d_i^2)^{k/2}} \right) \in (0, k/2).
$ 
In particular, for seeing that $
\phi^{-1} \left( \frac{2}{(1-w_i^2 d_i^2)^{k/2}} \right) >0,
$ 
we proceed a proof by contradiction by supposing that~$\phi^{-1} \left( \frac{2}{(1-w_i^2 d_i^2)^{k/2}} \right) \leq 0.$  
Taking $\phi(\cdot)$ on both sides and applying the strict increasingness of $\phi$, it leads to $\frac{2}{(1-w_i^2 d_i^2)^{k/2}} \leq  2$, which implies that 
\begin{align} \label{ineq: contradicting ineq}
	1 \leq (1- w_i^2 d_i^2)^{k/2}. 
\end{align}
However, since $w_i \in (0, 1)$ and $d_i \in (-1, 0)$, it follows that $(1 - w_i^2 d_i^2) \leq 1$. Raising both sides to the $k/2$ power leads to $(1 - w_i^2 d_i^2)^{k/2} < 1$, which contradicts to the fact~(\ref{ineq: contradicting ineq}). Therefore, we must have~$ \phi^{-1} \left( \frac{2}{(1-w_i^2 d_i^2)^{k/2}} \right) > 0. $
Lastly, for seeing that $\phi^{-1}\left( \frac{2}{(1-w_i^2 d_i^2)^{k/2}} \right) < k/2$, we must show that
\[
\frac{2}{(1-w_i^2 d_i^2)^{k/2}} < \phi(k/2) = \left( \frac{1 - w_i d_i}{1+ w_i d_i} \right)^{k/2}   +  \left( \frac{1 - w_i d_i}{1+ w_i d_i} \right)^{-k/2}.
\]
It is equivalent to show that
$
(1-w_i^2 d_i^2)^{k/2} \left( \frac{1 - w_i d_i}{1+ w_i d_i} \right)^{k/2}   +  (1-w_i^2 d_i^2)^{k/2} \left( \frac{1 - w_i d_i}{1+ w_i d_i} \right)^{-k/2} > 2.
$
Simplifying the expression, we obtain
$
(1 - w_i d_i)^{k} + (1+ w_i d_i)^{k} > 2 
$
which holds true for any $w_i d_i \neq 0$ for all $i$ and all~$k>0.$
\end{proof}

\medskip
\begin{remark} \rm
It is essential to note that the condition on $\varepsilon > \phi^{-1}(\cdot)$ in Theorem~\ref{theorem: RPE in Symmetric Lattice Market} is not an overly restrictive condition, and it can be easily satisfied. For illustration, within a daily basis scale,~$|d_i| \ll 1$, with one common value being $d_i \approx -0.02$, as can be seen in Section~\ref{section: empirical studies}. If we consider using multi-double linear policies and trading with one year comprising $k=252$ days, then the ratio~$2/(1-w_i^2 d_i^2)^{k/2} \geq 2$ for all $w_i \in (0, 1)$. 
Then the RPE holds if the condition $\phi(\varepsilon_i) \geq \frac{2}{(1-w_i^2 d_i^2)^{k/2}}$ in Theorem~\ref{theorem: RPE in Symmetric Lattice Market} is met for some $\varepsilon_i \geq 8$.  More details about these findings will be presented in Section~\ref{section: empirical studies} using historical data.
\end{remark}

\medskip
\section{Parameters Estimation} \label{section: parameters estimation}
In practice, the parameters of the generalized lattice model, as developed in previous sections, are typically unknown to the trader and must be estimated. This constitutes the primary goal of this section. Specifically, for each Asset~$i$ with $i= 1,2, \dots, n$, we shall discuss how to estimate the required parameters, e.g., $u_i, d_i, \Phi_{i,j}$, and $\Gamma_{i, \ell}$ for $j =0,1\dots,m$ and $\ell =1,2,\dots,n$.

\medskip
\subsection{Estimating the Upward and Downward Factors}
To estimate the upward and downward movement factors $u_i, d_i$ of returns, we work with a data-driven \textit{geometric mean} described as follows:
Let $X_i(j)$ be the per-period returns obtained from the data. 
Then, for each asset, we partition the return data into two separate series,  positive and negative series, denoted by $\{x_i^+(j)\}$ and~$\{x_i^-(j)\}$ consisting of $k^+$ positive values and $k^-$ negative values in the series, respectively.
Now we solve $\prod_{j=1}^{k^+} (1 + x_i^+(j)) = (1 + u_i)^{k^+}$ to find the upward movement estimate for Asset~$i$, denoted by $u_i = \widehat{u}_i$. Likewise, by solving $\prod_{j = 1}^{k^-}(1+x_i^-(j)) = (1 + d_i)^{k^-}$, we get the downward movement estimate, denoted by $d_i = \widehat{d}_i$, for Asset $i$.\footnote{Given $X_i = (X_i(1),\dots, X_i(k))$, the geometric mean of $X$, call it $\overline{X}_i$, is the solution of the nonlinear equation $
\prod_{j=0}^{k-1} (1+X_i(j)) = (1 + \overline{X}_i )^k$; for more details on this topic, we refer to \cite{casella2021statistical}. On the other hand, one can also compute the upward and downward movement factors $u_i, d_i$ by the standard  Cox-Ross-Rubinstein (CRR) method in finance; e.g., see \cite{cox1979option, luenberger2013investment}.}

\medskip
\subsection{Estimating the Correlation Matrix}
We estimate the correlation matrix, denoted by $\Gamma = \widehat{\Gamma}$, among the assets by determining the correlation coefficients of their realized asset returns. To achieve this, we collect $l$ historical returns data points for each Asset $i$ at stage $k=0$, denoted by~$x_i^*(-1:-l):=\{x_i^*({-1}), \dots,x_i^*({-l})\}$.
Specifically, for two distinct Assets $i$ and~$j$ where $i, j \in \{1, \dots, n\}$ and for $ i \neq j$, the correlation is represented by~$
\widehat{\Gamma}_{i,j} = {\rm corr}(x_i^*(-1:-l), x_j^*(-1:-l)).
$

Notably, in contrast to the conventional setting where the diagonal elements of the correlation matrix are equal to one, we define
$\Gamma_{i, i} := 0$ for $i \in \{1, \dots, n \}.$
The departure from tradition is motivated by the fact that the correlation of the previous stage is characterized by the Markov coefficient $\Phi_{i,1}$, which will be estimated in Section~\ref{subsection: estimating the Markov coefficients}.  
Hence, the estimated matrix $\widehat{\Gamma}$ can be written as
\begin{align*}
\widehat{\Gamma} := 
\begin{bmatrix}
	0 & \widehat{\Gamma}_{1,2} & \cdots & \widehat{\Gamma}_{1,n} \\
	\widehat{\Gamma}_{2,1} & 0 & \cdots & \widehat{\Gamma}_{2,n} \\
	\vdots & \vdots & \ddots & \vdots \\
	\widehat{\Gamma}_{n,1} & \widehat{\Gamma}_{n,2} & \cdots & 0
\end{bmatrix}.
\end{align*}

\medskip
\subsection{Estimating the Markov Coefficients}
\label{subsection: estimating the Markov coefficients}
Having obtained the estimates of $(u_i, d_i)$ for $i=1,\dots,n$, we then estimate the Markov coefficients~$\Phi_{i,j}$ for the probability model described in Section~\ref{section: Preliminaries}. 
Indeed, for Asset~$i$, we define the \textit{response} variable 
$
Y_i(k) := \frac{X_i(k) - d_i}{u_i - d_i}.
$
Then, it follows that 
\begin{align*}
\mathbb{E}\left[ Y_i(k) \mid X({k-1}: {k-m}) \right]  
& = \frac{1}{u_i - d_i} \left(\mathbb{E}[X_i(k) \mid X({k-1}: k-m)] - d_i \right) \\
& = \frac{u_i \left(  \Phi_{i,0} + \sum_{j=1}^m \Phi_{i,j} X_i(k-j) + \sum_{\ell=1}^n \Gamma_{i, \ell} X_\ell(k-1) \right) }{u_i - d_i} \\
&  \qquad 			+ \frac{ d_i \left(1- \Phi_{i,0}  -\sum_{j=1}^m \Phi_{i,j} X_i(k-j) - \sum_{\ell=1}^n \Gamma_{i, \ell} X_\ell(k-1) \right) - d_i}{u_i - d_i} \\
& = \Phi_{i,0} + \sum_{j=1}^m \Phi_{i,j} X_i(k-j) + \sum_{\ell=1}^n \Gamma_{i, \ell} X_\ell(k-1),
\end{align*}
where the second last equality holds by using Equality (\ref{eq: conditional probability for asset i}).

To estimate the parameters $\Phi_{i, j}$ for Asset $i$,
we collect $l$ historical returns data points, $x_i^*(-1:-l) = \{x_i^*(-1), \dots, x_i^*(-l)\}$, at stage $k=0$ satisfying $l >m$
, and transform it into binomial values~$\{x_i({-1}), \dots, x_i({-l}) \}$ based on the sign of per-period returns of each stage, which satisfy
$
x_{i}(j):= u_i \cdot\mathbb{1}_{\{x_{i}^*(j) \geq 0\}} + d_i \cdot \mathbb{1}_{\{x_i^*(j) < 0\}}
$
where $\mathbb{1}_A$ is the indicator function of event $A$.
Subsequently, we compute the response values $y_i(k) = \frac{1}{u_i - d_i}(x_i(k) - d_i)$ and minimize the \textit{residual sum of squares} (RSS)~as 
\[
{\rm RSS}(\Phi_{i, 0}, \Phi_{i,1}, \dots, \Phi_{i, m}) := \sum_{k = -l+m}^{-1}\left( y_i(k) - \Phi_{i, 0} - \sum_{j = 1}^{m} \Phi_{i,j} x_i({k-j})  -  \sum_{\ell=1}^n \Gamma_{i, \ell} x_\ell(k-1) \right)^2.
\]
Next lemma shows that the parameter space of Markov coefficients forms a convex polyhedron, which facilitates the optimization.

\medskip
\begin{lemmarep}[Convex Constraint Set on Markov Coefficients] \label{lemma: markov chain conditional prob model}
Let $m, n \geq 1$.  For $i=1,\dots,n$, the inequalities $0 \leq \Phi_{i,0} + \sum_{j=1}^{m} \Phi_{i,j} x_{i}(k-j) + \sum_{\ell=1}^{n} \Gamma_{i,\ell} x_{\ell} (k-1)\leq1$ where $x_{i}(\cdot) \in \{ u_i, d_i\} $ are equivalent~to
\begin{align} \label{eq: prob constraints}
	\left|\Phi_{i, 0} - \frac{1}{2} + \frac{u_i + d_i}{2}  \sum_{j=1}^{m} \Phi_{i,j} +\sum_{\ell=1}^{n} \frac{u_\ell + d_\ell}{2}\Gamma_{i,\ell}\right| 
	+ \frac{u_i - d_i}{2} \sum_{j=1}^{m} \left|  \Phi_{i, j} \right|
	+ \sum_{\ell=1}^{n} \frac{u_\ell-d_\ell}{2} \left|\Gamma_{i,\ell}\right|\leq \frac{1}{2},
\end{align}
which forms a convex polyhedron.
\end{lemmarep}

\medskip
\textit{Proof.} See Appendix~\ref{Appendix: Technical Proofs of Section: gain-loss analysis}.

\medskip
For each Asset~$i$, with the aid of Lemma~\ref{lemma: markov chain conditional prob model}, the optimal estimators, call it $\widehat{\Phi}_{i, 0}, \widehat{\Phi}_{i,1}, \dots, \widehat{\Phi}_{i, m}$, can be solved by the following \textit{constrained least squares optimization} problem:
\begin{align}
\label{eq: RSS constraint problem}
& \min_{\Phi_{i, 0}, \Phi_{i,1}, \dots, \Phi_{i, m}} \quad {\rm RSS}(\Phi_{i, 0}, \Phi_{i,1}, \dots, \Phi_{i, m}) \nonumber\\
&{\rm s.t.} \;   
\left|\Phi_{i, 0} - \frac{1}{2} + \frac{u_i + d_i}{2}  \sum_{j=1}^{m} \Phi_{i,j} +\sum_{\ell=1}^{n} \frac{u_\ell + d_\ell}{2}\Gamma_{i,\ell}\right|  + \frac{u_i - d_i}{2} \sum_{j=1}^{m} \left|  \Phi_{i, j} \right|
+ \sum_{\ell=1}^{n} \frac{u_\ell-d_\ell}{2} \left|\Gamma_{i,\ell}\right|\leq \frac{1}{2}, \nonumber
\end{align}
Note that the optimization problem above is a linear constraint quadratic convex program; see \cite{boyd2004convex}, which facilitates an efficient computation.

\medskip
\subsection{Estimating the Time-Varying Probabilities}
Recalling that in Lemma~\ref{lemma: recursion formula of probability}, the probability  of receiving a positive return for Asset~$i$ at stage~$k$, i.e., $p_i(k) = \mathbb{P} (X_i(k) = u_i)$ can be computed by the following recursive equation
\begin{align}  
p_i (k) 
&= \Phi_{i,0} + \sum_{j=1}^m \Phi_{i,j} [(u_i - d_i) \cdot p_i({k-j}) + d_i]+ \sum_{\ell=1}^n \Gamma_{i, \ell} [(u_\ell - d_\ell) \cdot p_\ell({k-1}) + d_\ell].
\end{align}
Later in Section \ref{section: empirical studies}, this recursion formula~(\ref{eq: p_k recursion}) is useful to estimate the binomial probability for each~stage.


\medskip
\section{Empirical Studies}
\label{section: empirical studies}
In this section, we conduct extensive empirical studies using historical price data from the top~30 holdings of the S\&P~500, sorted by market capitalization. In the sequel, we shall refer to these~$30$ assets as ``S\&P 30."\footnote{The thirty stocks used in this paper, in terms of their tickers, are: \textbf{AAPL} (Apple Inc.), \textbf{MSFT} (Microsoft Corporation), \textbf{AMZN} (Amazon.com Inc.), \textbf{NVDA} (NVIDIA Corporation), \textbf{GOOGL} (Alphabet Inc. Class~A), \textbf{GOOG} (Alphabet Inc. Class C), \textbf{META} (Meta Platforms Inc.), \textbf{BRK.B} (Berkshire Hathaway Inc. Class B), \textbf{TSLA} (Tesla Inc.), \textbf{UNH} (UnitedHealth Group Inc.), \textbf{JNJ} (Johnson \& Johnson), \textbf{JPM} (JPMorgan Chase \& Co.), \textbf{LLY} (Eli Lilly and Company), \textbf{XOM} (Exxon Mobil Corporation), \textbf{V} (Visa Inc.), \textbf{PG} (Procter \& Gamble Co.), \textbf{AVGO} (Broadcom Inc.), \textbf{HD} (The Home Depot Inc.), \textbf{MA} (Mastercard Incorporated), \textbf{CVX} (Chevron Corporation), \textbf{MRK} (Merck \& Co.), \textbf{ABBV} (AbbVie Inc.), \textbf{PEP} (PepsiCo Inc.), \textbf{COST} (Costco Wholesale Corporation), \textbf{ADBE} (Adobe Systems Incorporated), \textbf{KO} (The Coca-Cola Company), \textbf{WMT} (Walmart Inc.), \textbf{CSCO} (Cisco Systems Inc.), \textbf{MCD} (McDonald's Corporation), and \textbf{BAC} (Bank of America Corporation). The price data were retrieved from Yahoo! Finance and the associated estimated parameters are provided in Appendix~\ref{appendix: Estimated Parameters of the S&P30 Stocks}.} 
Our focus is twofold: The first is to demonstrate the efficacy of the generalized lattice market; the second is to illustrate the robust positive expected profits of the multi-double linear policies. 
Additionally, we include a comprehensive sensitivity analysis to further enhance our contributions. 
Toward the end of this section, a method for identifying optimal weights is discussed, accompanied by several studies of the out-of-sample trading performance of the policies.

\medskip
\subsection{Validation of the RPE Property}
We validate the RPE property through two examples: The first example confirms the RPE versus weights; the second investigates the property via sensitivity analysis; i.e., we explore the effects of various parameters such as weights, memory length, and initial allocation constants.

\medskip
\begin{example}[RPE of S\&P 30 Portfolio]\rm
\label{example: RPE with SP30 Portfolio}
To validate our theoretical results, we consider the~S\&P~30 portfolio consisting of top $n=30$ stocks listed on the S\&P 500 index. 
For each stock, we collect the historical daily adjusted closing prices over a one-year period from January 2022 to December 2022. 
Our analysis begins with the estimation of essential parameters of each asset, using the methods detailed in Section~\ref{section: parameters estimation}.
Specifically, we estimate the upward and downward movement factors $\widehat{u}_i$, $\widehat{d}_i$, Markov coefficients~$\widehat{\Phi}_{i,j}$, correlation coefficients~$\widehat{\Gamma}_{i,\ell}$, and the probability $p_i(j)$ for each asset. The detailed values of these estimates can be found in Appendix~\ref{appendix: Estimated Parameters of the S&P30 Stocks}.

\medskip
\textit{Efficacy of Generalized Lattice Market Model.}
With these estimations in hand, we proceeded to simulate the stock prices using the recursion~$S_i(k+1) = S_i(k) (1+X_i(k))$ for $i \in \{1, 2, \dots,30\}$ over a time horizon of $252$ days, covering the period from January 2023 to July 2023, with~$10,000$ sample paths. 
Figure~\ref{fig: Monte Carlo simulation} reveals both the actual historical stock prices (in solid black lines) for 2022 and the simulated 500 paths for the first seven months of 2023. 
The paths were generated using the aforementioned parameter estimates. Remarkably,  an examination of the figure illustrates that the simulated paths approximately encompass the real stock prices for the corresponding period in~2023, even under varying market conditions. 
This figure indicates the potential of our generalized lattice model to offer insights into market behavior, demonstrating that the training data might represent a bear market, while the testing data may indicate a bull market.

\begin{figure}[h!]
	\centering
	\includegraphics[width=.7\linewidth]{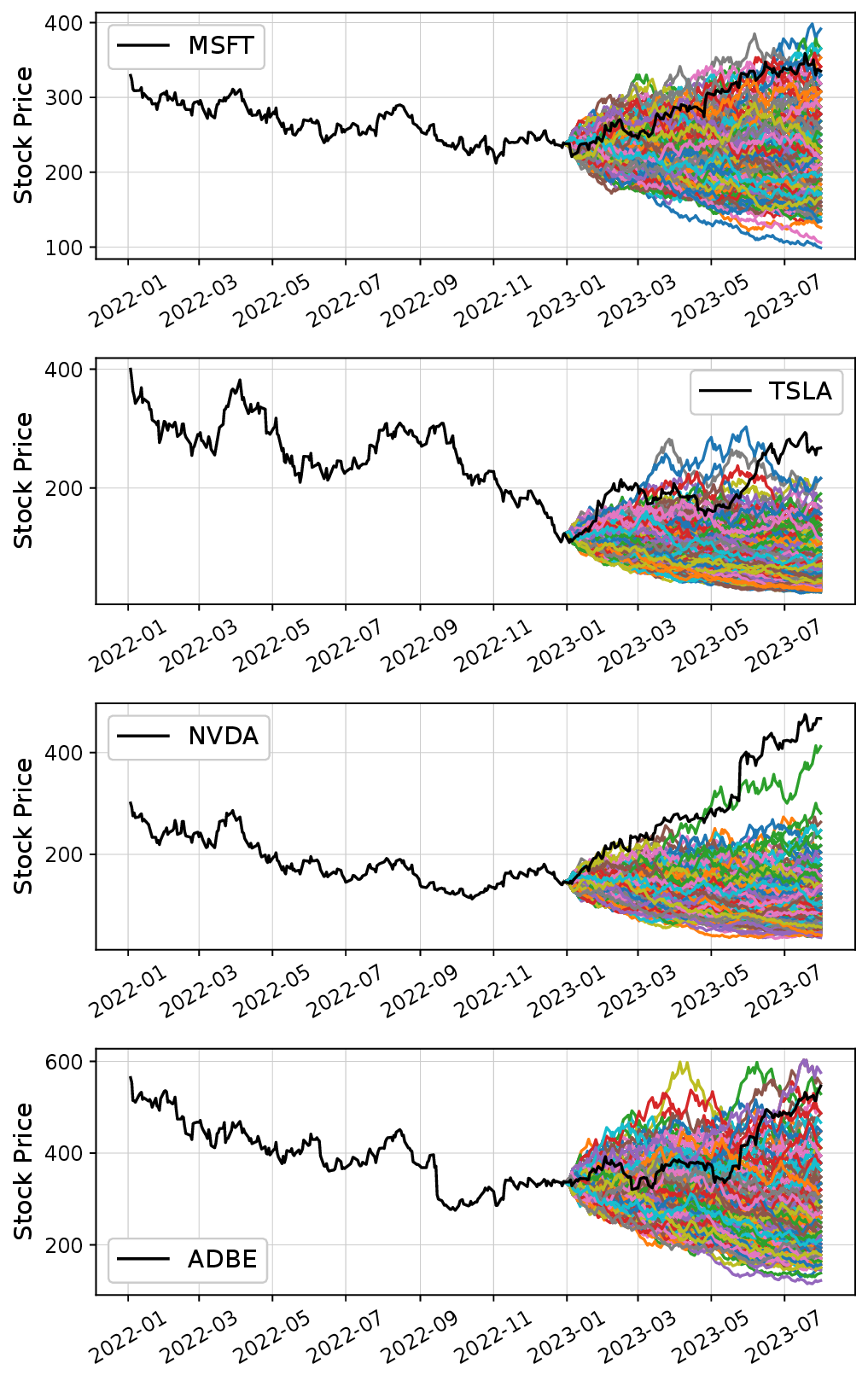}
	\caption{Example of Real Stock Prices in S\&P 30 Versus Monte Carlo Simulated Sample Paths Using Generalized Lattice Model with Memory Length $m=5$.}
	\label{fig: Monte Carlo simulation}
\end{figure}

\medskip
\textit{Initializing the Multi-Double Linear Policies.}
As mentioned in previous sections, to implement the multi-double linear policies, one must specify the triple~$(\alpha, w, v)$ and initial account value~$V_0 > 0$.
Here, we first take $\alpha = 1/2$, an initial account value of~$V_0 = \$1$, and equal asset weights $w_i = \omega \in [0, 1]$ for all $i=1,2, \dots, 30$ for S\&P 30, using a generalized lattice model with varying memory lengths~$m \in \{1, 2, 5, 10\}$ and a zero risk-free rate of $r_f = 0$.
For the specification of the initial allocation vector $v$, we consider three distinct strategies: The first one is \textit{equal-weighted} allocation, denoted as~$v_{EW} = [1/30\; \cdots \;1/30]^\top$. 
The second one is \textit{capital-weighted} allocation, computed by the normalized weight of the S\&P 30 index, denoted as~$v_{{\rm CW}}$.\footnote{The capital-weighted allocation is retrieved from Slickcharts.}
The third one is \textit{gain-loss-weighted} allocation, calculated using the gain-loss of training data as a criterion for capital allocation. It is denoted as $v_{GL, i} := |S_i(k)/S_i(0) - 1 |/\sum_{i=1}^{n} |S_i(k)/S_i(0) - 1 |$, forming the vector $v_{\rm GL}$.
Table \ref{tab: Summary of Parameters Sets} summarizes the choice of the relevant parameters.

\begin{table}[h!]
	\centering
	\caption{Summary of Parameters Sets.}
	\begin{tabular}{c|c c c}
		& $(i)$ & $(ii)$ & $(iii)$ \\
		\hline
		$V_0~(\$)$ & 1 & 1 & 1 \\
		$r_f~(\%)$ & 0 & 0 & \{0.92, 1.51, 3.88\} \\
		$m$ & \{1, 2, 5, 10\} & 1 & 1 \\
		$\alpha$ &  0.5 & \{0.1, 0.3, 0.5, 0.7, 0.9\} & 0.5 \\
		$w$ & $\omega \in [0,1]$ & $\omega \in [0,1]$ & $\omega \in [0,1]$ \\
		$v$ & $\{v_{\rm EW}, v_{\rm CW}, v_{\rm GL}\}$ & $v_{\rm EW}$ & $v_{\rm EW}$\\
	\end{tabular}
	
	\label{tab: Summary of Parameters Sets}
\end{table}

\medskip
\textit{Gain-Loss Analysis.}
With the various settings of the triple $(\alpha, w, v)$ as seen in Table~\ref{tab: Summary of Parameters Sets}, we implement the multi-double linear policies across the $30$ assets within the portfolio.
The trading performance, as illustrated in Figure~\ref{fig: expected gain-loss against weights}, depicts the average cumulative gain-loss against various weights for the parameter set $(i)$ in Table~\ref{tab: Summary of Parameters Sets}. 
We see that there exists a positive expected gain-loss~$\overline{\mathcal{G}}(k) \geq 0$ for the given portfolio and for all parameter set combinations.
Remarkably, this positivity is found to be robust across diverse choices of $v, w, m$ at~$\alpha = 1/2$.

\begin{figure}[h!]
	\centering
	\includegraphics[width=.8\linewidth]{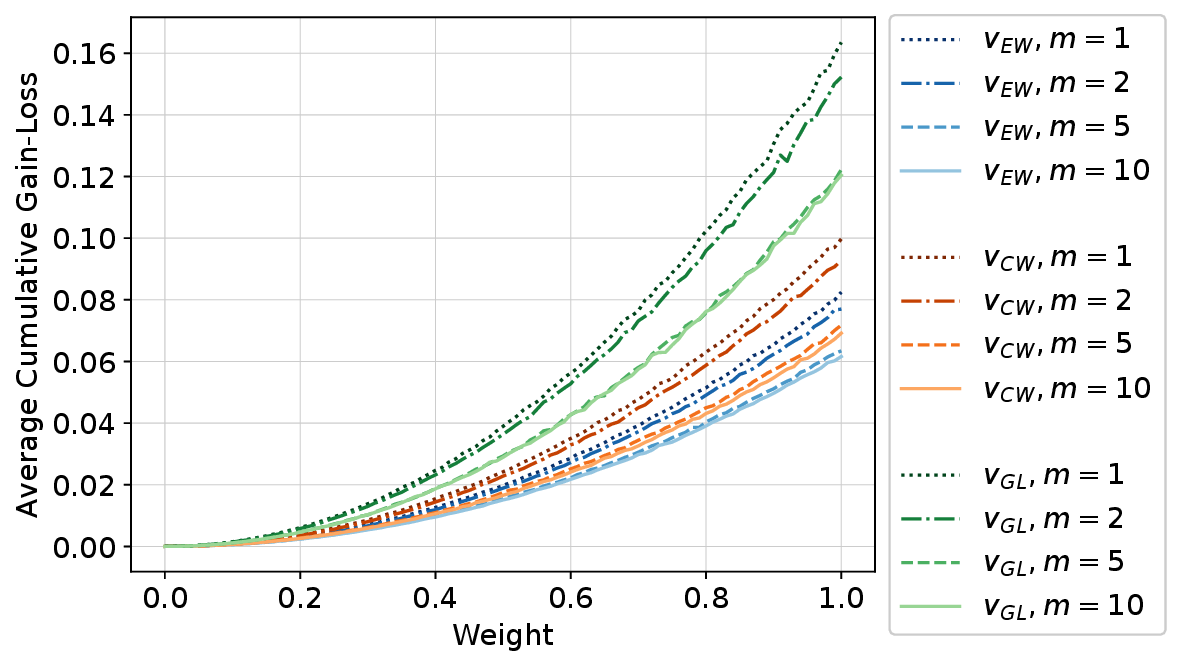}
	\caption{Average Cumulative Gain-Loss $\overline{\mathcal{G}}(k)$ Against Weights $\omega \in [0,1]$ for Parameter Set $(i)$ Described in Table~\ref{tab: Summary of Parameters Sets}.}
	\label{fig: expected gain-loss against weights}
\end{figure}

Moreover, we present the associated standard deviation of the cumulative gain-loss for the parameter set $(i)$, depicted in Figure \ref{fig: variance of gain-loss against weights}. 
Notably, both the average cumulative gain-loss and the corresponding standard deviation exhibit increasing monotonicity with respect to the weights, as shown in Figures \ref{fig: expected gain-loss against weights} and \ref{fig: variance of gain-loss against weights}.
The results presented here form a foundation for practitioners aiming to explore novel trading paradigms; see Section~\ref{subsection: Trading Performance with Transaction Costs} to follow.

\begin{figure}[h!]
	\centering    \includegraphics[width=.8\linewidth]{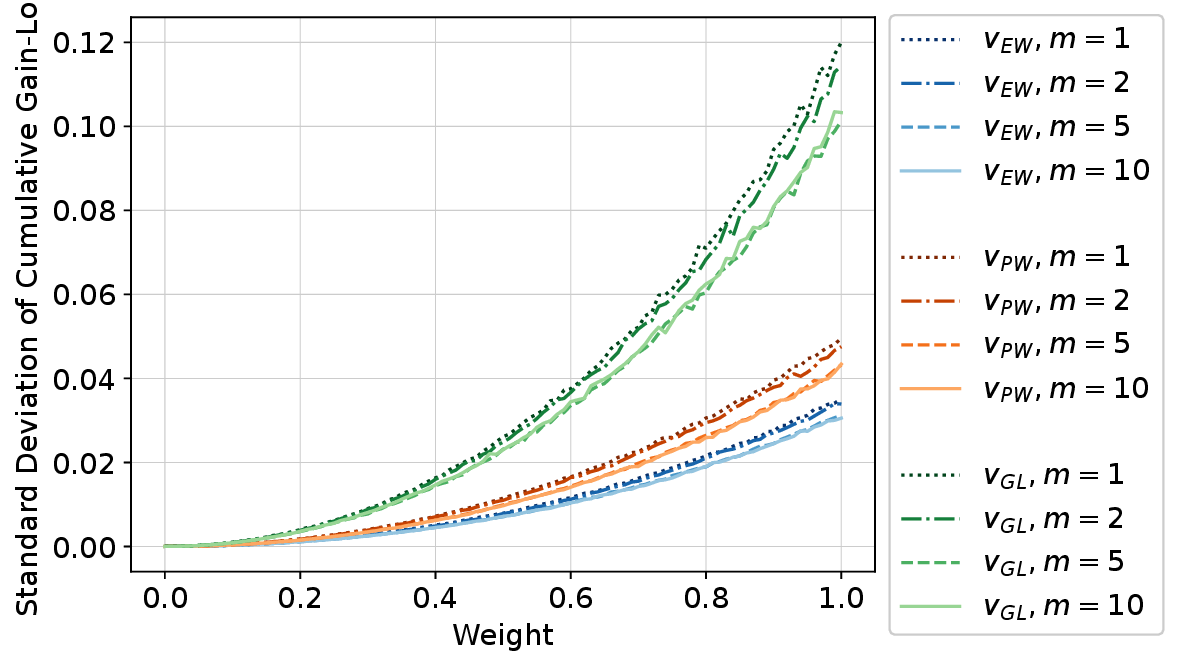}
	\caption{Standard Deviation of Cumulative Gain-Loss ${\rm std}(\mathcal{G}(k))$ Against Weights $\omega \in [0, 1]$ for Parameter Set~$(i)$ Specified in Table~\ref{tab: Summary of Parameters Sets}.}
	\label{fig: variance of gain-loss against weights}
\end{figure}

\end{example}

\medskip
\begin{example}[Sensitivity Analysis] \rm
Continuing with the dataset and estimations provided in Example~\ref{example: RPE with SP30 Portfolio}, we proceed further to analyze the gain-loss trading performance under different~$\alpha \in \{0.1, 0.3, 0.5, 0.7, 0.9\}$, in combination with various weights~$\omega$, as the parameter set $(ii)$ listed in Table~\ref{tab: Summary of Parameters Sets}. 
The corresponding results are shown in Figure~\ref{fig: gain-loss against weights for different alpha}. 
We see that the average gain-loss is positive for $\alpha \in \{0.1, 0.3, 0.5\}$, implying that there is a preference for initially allocating more capital to the short positions. 
In contrast, a negative average gain-loss is recorded for $\alpha \in \{0.7, 0.9\}$. This phenomenon might be attributable to a decreasing trend for the S\&P 30 over the examined period.

\begin{figure}[h!]
	\centering
	\includegraphics[width=.8\linewidth]{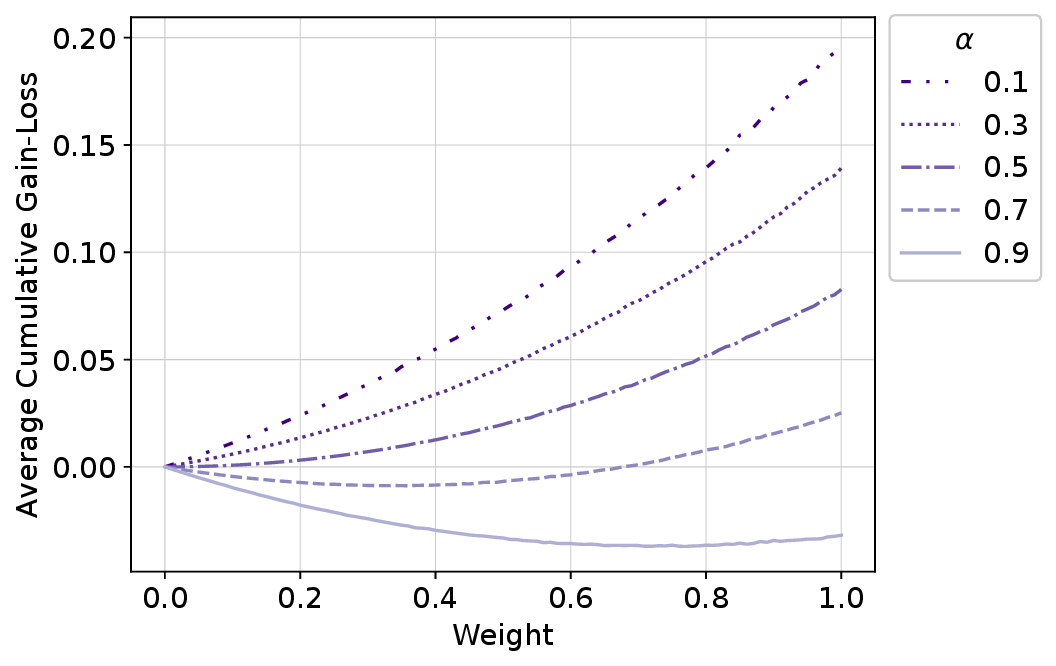}
	\caption{Sensitivity of $\alpha$: Average Cumulative Gain-Loss Versus Weights Under Different $\alpha$ Values.}
	\label{fig: gain-loss against weights for different alpha}
\end{figure}

Besides, we extend our analysis to implement multi-double linear policies that incorporate a non-zero risk-free rate, $r_f \in \{0.92\%, 1.51\%, 3.88\%$\},\footnote{The three risk-free rates $r_f$ listed here were retrieved from the U.S. 10-year Treasury yields on the first day of the years 2021, 2022, and 2023, respectively.} as outlined in the parameter set $(iii)$ in Table \ref{tab: Summary of Parameters Sets}. 
By comparing the performances under parameter set~$(i)$ with  $v_{EW}$, $m=1$ versus $(iii)$, represented respectively in Figure \ref{fig: non-zero risk-free rate}, we observe that incorporating a risk-free rate for long positions leads to a higher cumulative gain-loss, which is consistent with the intuition that the risk-free rate would contribute additional capital due to under-investment in each period.

\begin{figure}[h!]
	\centering
	\includegraphics[width=.8\linewidth]{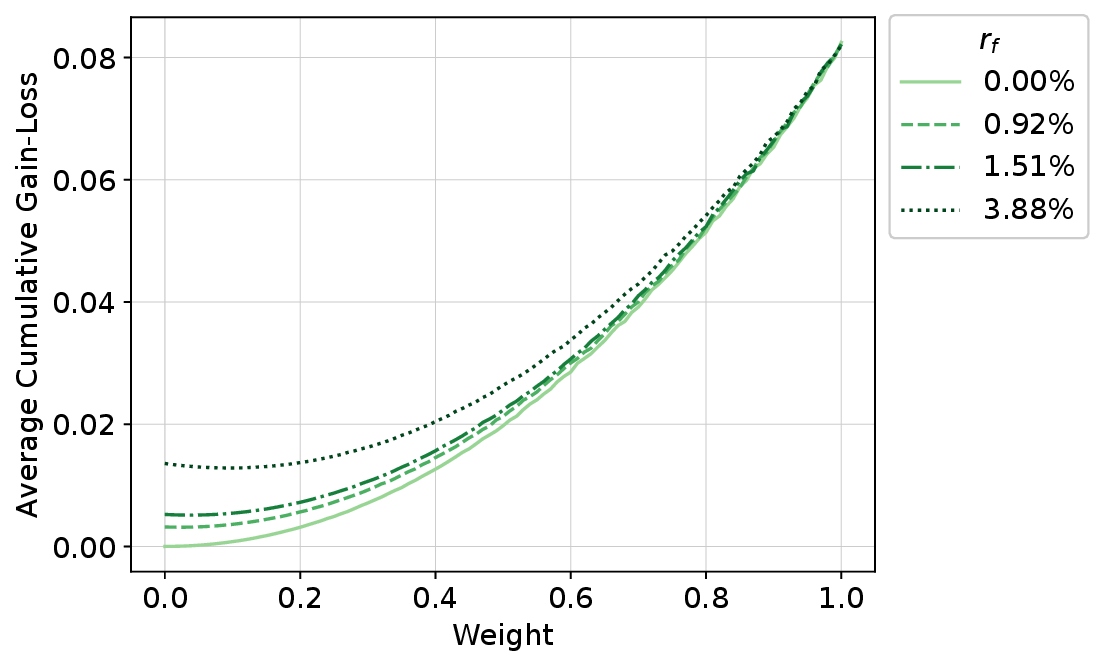}
	\caption{Sensitivity of $r_f$: Average Cumulative Gain-Loss Versus Weight Under Different Risk-Free Rates.}
	\label{fig: non-zero risk-free rate}
\end{figure}

\end{example}

\medskip
\subsection{Out-of-Sample Trading Performance with Costs} \label{subsection: Trading Performance with Transaction Costs}
This section studies out-of-sample trading performance using the multi-double linear policies with triple $(\alpha, w, v) \in \{1/2\} \times [0,1]^n \times \mathcal{V}$. In the following analysis, we also consider the impact of transaction costs.

\medskip
\subsubsection{Identifying Optimal Weights via Mean-Variance Efficient Frontier}
When the weight $w_i = \omega$ for all $i=1, \dots, n$,  Figure~\ref{fig: expected gain-loss against weights} reveals that the expected cumulative gain-loss increases as the weight~$\omega$ increases. 
This is also evident in the standard deviation of the cumulative gain-loss across different weights; see Figure~\ref{fig: variance of gain-loss against weights}.
Therefore, we then depict the mean-variance efficient frontiers using the expected cumulative gain-loss~$\overline{\mathcal{G}}(k)$ against the standard deviation ${\rm std}(\mathcal{G}(k))$, as shown in Figure~\ref{fig: standard deviation against expected gain-loss}.
Based on the monotonic increasing property, 
one can determine an \textit{optimal} weight, call it $w^*$, using the mean-variance approach outlined in~\cite{hsieh2022robust}.

\begin{figure}[h!]
\centering
\includegraphics[width=.8\linewidth]{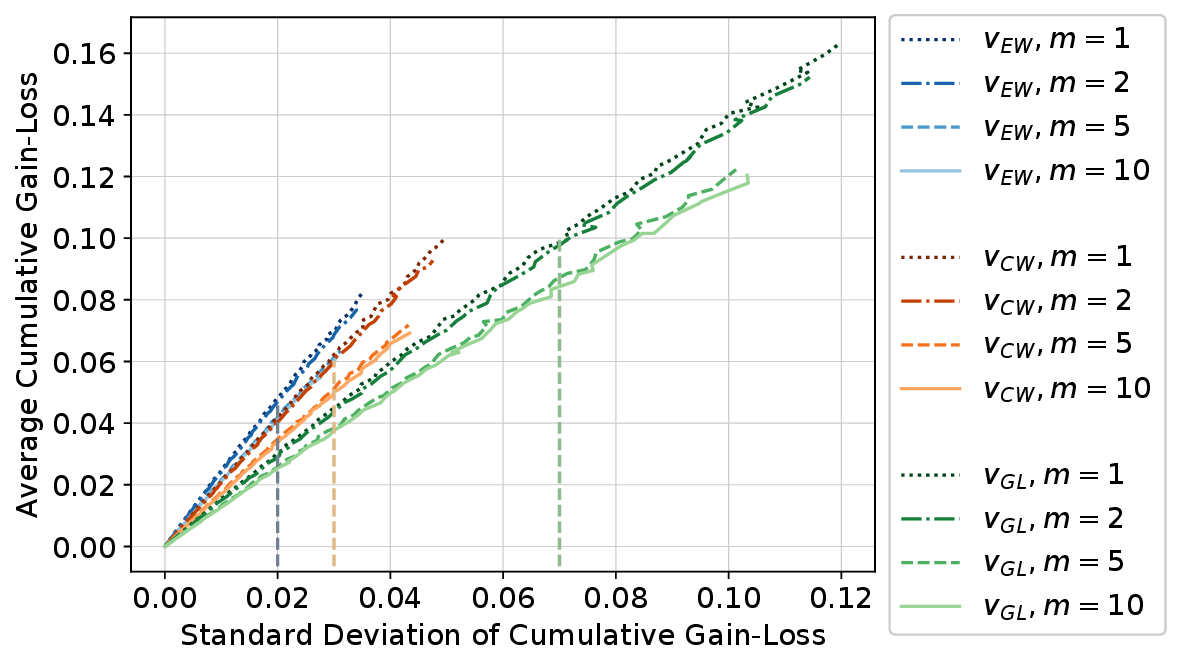}
\caption{Mean-Standard Deviation Efficient Frontiers: Expected Cumulative Gain-Loss $\overline{\mathcal{G}}(k)$ Against Standard Deviation ${\rm std}(\mathcal{G}(k))$ for Parameter Set $(i)$.}
\label{fig: standard deviation against expected gain-loss}
\end{figure}

In particular, the procedure is as follows: We begin by specifying a \textit{target} standard deviation,  then identify the corresponding maximum expected gain-loss under the chosen target. 
Once this maximum expected gain-loss is decided, it can be traced back to the optimal weight due to the strictly increasing relationship between expected gain-loss and weight.
Below we consider three different approaches to identifying optimal weights.

\medskip
\textit{Constant Optimal Weight for S\&P 30.}
Consider the same S\&P 30 dataset from Example~\ref{example: RPE with SP30 Portfolio}.
Suppose we set the target standard deviation ${\rm std}(\mathcal{G}_{\rm EW}(k)) := 0.02$, ${\rm std}(\mathcal{G}_{\rm CW}(k)) := 0.03$, and ${\rm std}(\mathcal{G}_{\rm GL}(k)) := 0.07$ for $v_{\rm EW}$, $v_{\rm CW}$, and $v_{\rm GL}$, respectively, represented by the dashed line in Figure~\ref{fig: standard deviation against expected gain-loss}, for the parameter set $(i)$ with memory length $m=1$. 
The maximum expected gain-loss are determined as $\overline{\mathcal{G}}_{\rm EW}(k) \approx 0.0479$, $\overline{\mathcal{G}}_{\rm CW}(k) \approx 0.0614$, and $\overline{\mathcal{G}}_{\rm GL}(k) \approx 0.0963$. 
Due to the strictly monotonic increasing relationship between expected gain-loss and weight $\omega$, as seen in Figure \ref{fig: expected gain-loss against weights}, we can readily trace back the optimal weight: $w^*_{c, \rm EW} \approx 0.77$, $w^*_{c, \rm CW} \approx 0.79$, and  $w^*_{c, \rm GL} \approx 0.78$, collectively denoted as $w_c^* := \{w^*_{c, \rm EW}, w^*_{c, \rm CW}, w^*_{c, \rm GL} \}$.

\medskip
\textit{Optimal Weight Determined by Top 10 Companies of S\&P 30.}
Except for assigning the same constant weight $w_c^*$ for all assets, an alternative approach focuses solely on investing assets that exhibit the most promising performance.
Here, we consider allocating the optimal constant weights to the top 10 assets of S\&P 30 that exhibit the highest gain-loss, as determined by the training data. For the remaining assets that are not selected, we assign weights of zero. The corresponding weight is denoted as~$w_{\rm top}^*$.

\medskip
\textit{Optimal Weight Determined by Each Company.}
Lastly, we explore determining optimal weights for individual assets, denoted by $w_i^*$ for $i=1,\dots,n$, using the monotonic increasing property between gain-loss function and weight.
For instance, by setting a uniform target standard deviation~${\rm std}(\mathcal{G}(k)) := 0.01$ for all assets, represented by the dashed line shown in Figure \ref{fig: gbar v.s. std for each asset}, we can determine the associated maximum average cumulative gain-loss for each asset, leading to a set of optimal weights called~$w_{\rm vary}^*:=\{w_i^*\}$ via a similar method mentioned previously.\footnote{In the meantime, we also assign a zero weight to the assets with an average gain-loss close to zero, using a threshold of $0.0001$.}

\begin{figure}[h!]
\centering
\includegraphics[width=.8\linewidth]{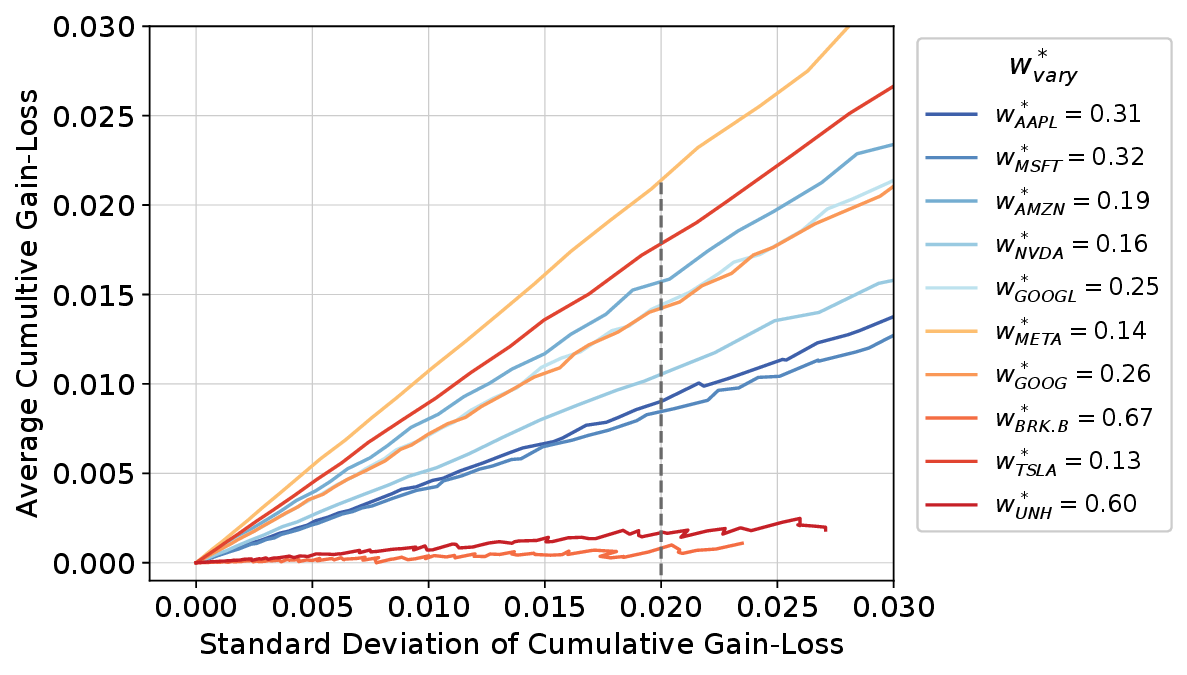}
\caption{Mean-Variance Efficient Frontiers: Average Cumulative Gain-Loss $\overline{\mathcal{G}}_i(k)$ Against Standard Deviation~${\rm std}(\mathcal{G}_i(k))$ for top 10 Asset in S\&P 30 Portfolio.}
\label{fig: gbar v.s. std for each asset}
\end{figure}

\medskip
\subsubsection{Out-of-Sample Trading  Performance via Backtesting}
With the selected optimal weights obtained in the preceding sections, we conducted backtesting over an out-of-sample period from January 2023 to July 2023. 
To analyze the trading performance across distinct sets of triple parameters $(\alpha, w, v)$, we carry out three initial allocations $v \in \{v_{\rm EW}, v_{\rm CW}, v_{\rm GL}\}$ and three types of optimal weights~$w \in \{w_{c}^*, w_{\rm top}^*, w_{\rm vary}^* \}$, the constant optimal weight, the top 10 constant optimal weight, and the optimal weight varying among assets, respectively at $\alpha = 1/2$.
The corresponding trading trajectories and performance are shown in Figure~\ref{fig: trading trajectory with distinct triple sets} and quantitatively summarized in Table~\ref{tab: trading performance with several triple parameter sets}.
For these scenarios, we incorporate a non-zero risk-free rate $r_f = 3.88\%$ and a transaction cost $c = 0.01\%$, corresponding to 1 basis point, for multi-double linear policies. 

Due to the explicit trend for the year~2023 for some assets listed in~S\&P 30, the multi-double linear policies assure positive gain-loss for several parameter sets even in the consideration of transaction costs.
Additionally, the triple~$(\alpha, w, v) = (1/2, w_{\rm top}^*, v_{\rm GL})$ has maximium out-of-sample gain-loss, $(1/2, w_{\rm top}^*, v_{\rm EW})$ has minimum standard deviation, and~$(1/2, w_{\rm vary}^*, v_{\rm GL})$ has minimum maximum percentage drawdown.

Additionally, Figure~\ref{fig: trading trajectory with distinct triple sets} also presents Monte-Carlo-based 95\% prediction intervals for each trading scenario in the figure. These intervals were generated via 10,000 simulated prices within the framework of the generalized lattice market model. For each path, we compute the average gain-loss and standard deviation for each day.
In the figure, three distinct prediction intervals are highlighted. The dark-gray, gray, and light-gray regions correspond to the parameter triples $(1/2, w_{\rm top}^*, \cdot)$, $(1/2, w_{\rm c}^*, \cdot)$, and $(1/2, w_{\rm vary}^*, \cdot)$, respectively.
These intervals exhibit significant positive gain-loss regions, which is consistent with our theory.

\begin{figure}[h!]
\centering
\includegraphics[width=.75\linewidth]{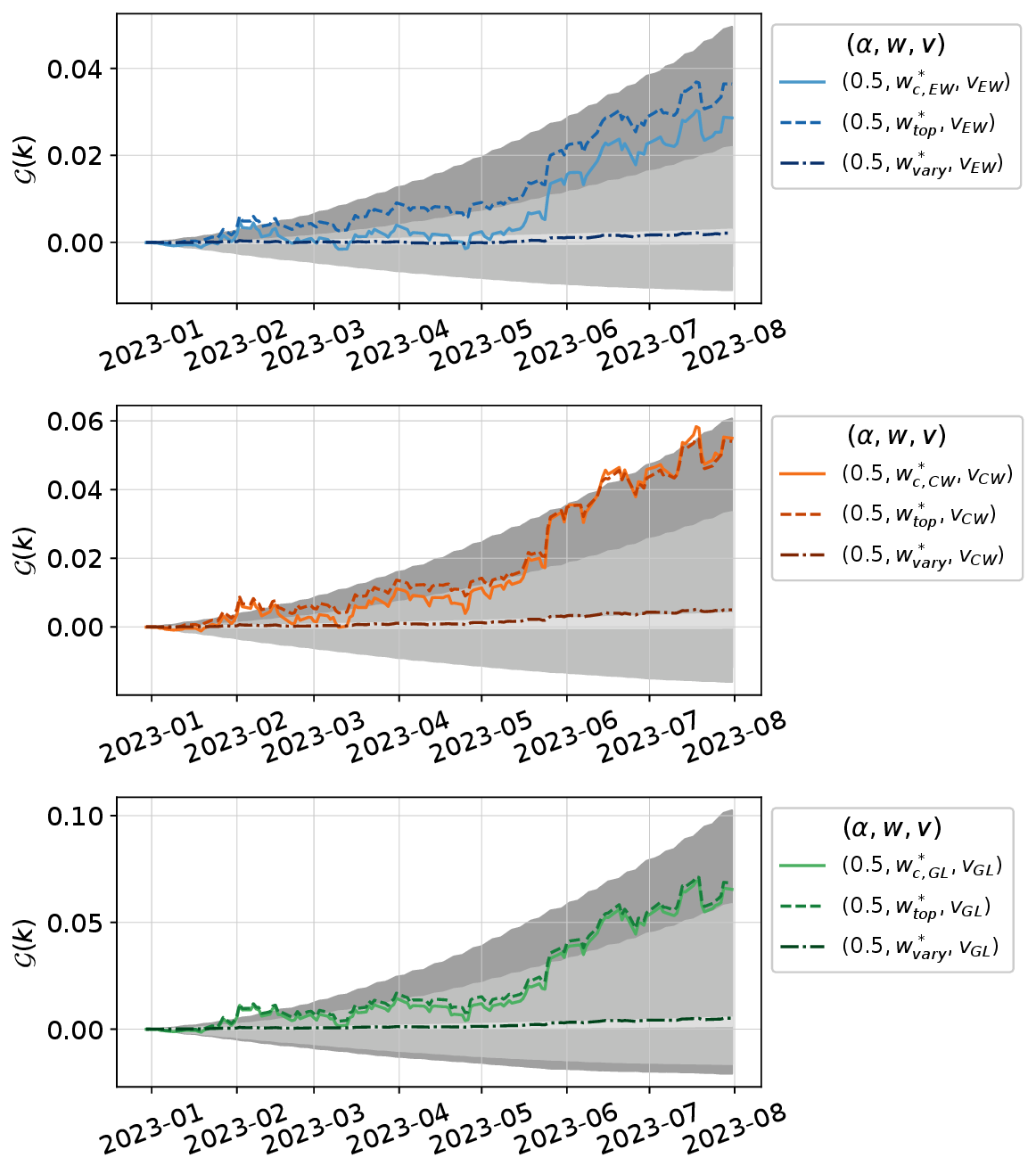}
\caption{Gain-Loss Trajectories under Multi-Double Linear Policies: The figure depicts the impact of various triple set parameters and transaction costs on the trajectories. Three distinct Monte-Carlo-based 95\% prediction intervals are highlighted: the dark-gray region corresponds to the triple $(1/2, w_{\rm top}^*, \cdot)$, the gray region to the  triple~$(1/2, w_{\rm c}^*, \cdot)$, and the light-gray region to the triple $(1/2, w_{\rm vary}^*, \cdot)$. }
\label{fig: trading trajectory with distinct triple sets}
\end{figure}

\begin{table}[h!]
\centering
\caption{Out-of-Sample Trading Performance of Multi-Double Linear Policies with Various Triples (Best Performance in Each Column Denoted in  Bold.}
\begin{tabular}{l|c c c}
	$(\alpha, w, v)$ & Gain-Loss & Standard Deviation & Maximum Drawdown (\%) \\
	\hline
	$(1/2, w_{c, \rm EW}^*, v_{\rm EW})$   &   0.0286    & 0.0099  &   0.6675  \\
	$(1/2, w_{\rm top}^*, v_{\rm EW})$     &   0.0364    & 0.0114  &  0.5964  \\
	$(1/2, w_{\rm vary}^*, v_{\rm EW})$    &   0.0021    & \textbf{0.0007}  &   0.0721  \\
	$(1/2, w_{c, \rm CW}^*, v_{\rm CW})$   &   0.0549    & 0.0186  &   1.0352  \\
	$(1/2, w_{\rm top}^*, v_{\rm CW})$     &   0.0540    & 0.0171  &   0.8278  \\
	$(1/2, w_{\rm vary}^*, v_{\rm CW})$    &   0.0050    & 0.0016  &   0.0737  \\
	$(1/2, w_{c, \rm GL}^*, v_{\rm GL})$   &   0.0654    & 0.0216  &   1.3483  \\
	$(1/2, w_{\rm top}^* , v_{\rm GL})$    &   \textbf{0.0683}   & 0.0219  &   1.2969  \\
	$(1/2, w_{\rm vary}^*, v_{\rm GL})$    &   0.0051    & 0.0016  &   \textbf{0.0592}    
\end{tabular}
\label{tab: trading performance with several triple parameter sets}
\end{table}

It is worth noting that the initial allocation for long and short positions, $\alpha$, influences the gain-loss performance; see Figure~\ref{fig: gain-loss against weights for different alpha}. 
As a result, we conduct the backtesting with varying~$\alpha \in \{0.1, 0.3, 0.5, 0.7, 0.9\}$ as demonstrated in Figure \ref{fig: trading trajectory with buy and hold}. 
Our analysis reveals that a positive gain-loss is attained for~$\alpha \in \{0.5, 0.7, 0.9\}$, while negative gain-loss is observed for~$\alpha \in \{0.1, 0.3\}$, considering a transaction cost of $0.01\%$ per trade and an annual risk-free rate of $r_f = 3.88\%$. This phenomenon in gain-loss can be attributed to the bullish market trends present in the testing data for the year~2023. 

\begin{figure}[h!]
\centering
\includegraphics[width=.8\linewidth]{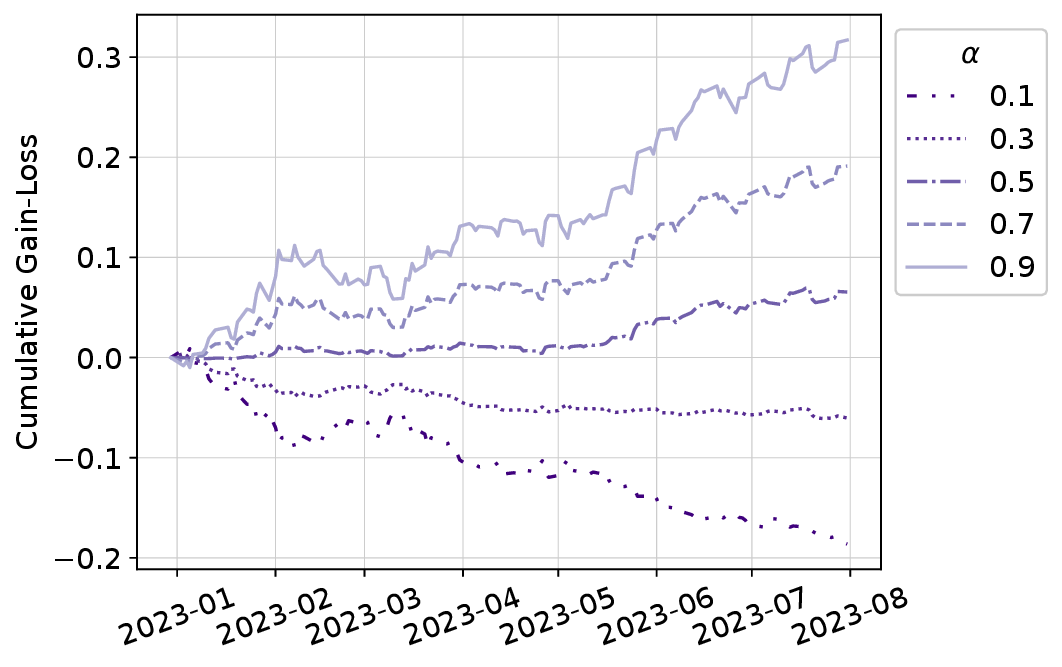}
\caption{Cumulative Gain-Loss Trajectories with Varying $\alpha$.}
\label{fig: trading trajectory with buy and hold}
\end{figure}

\section{Concluding Remark}
In this paper, we investigated the stochastic robustness of a novel trading scheme called the multi-double linear policy within a generalized lattice market that accounts for both serially correlated returns and asset correlation. 
We proved that our policies assure survivability and probabilistic positivity.
We derived a detailed gain-loss performance analysis, leading to the derivation of an analytic expression for the worst-case expected gain-loss and evidence that the proposed policies might sustain an ``approximate'' robust positive expected profit. 
Additionally, we proved sufficient conditions that multi-double linear policies can lead to positive expected profits if the market has a clear trend. 
This remarkable result holds true even within a symmetric lattice market.
Through a conditional probabilistic model, we uncovered that the parameter space of our mode forms a convex polyhedron, estimated efficiently via a constrained least-squares method. 
Validated by extensive empirical studies with S\&P 30 assets, we substantiate robust support for the theoretical findings.

As for future research, extending our analysis to the current findings to involve time-varying weights $w_i(k)$ with time-varying movement factors $u_i(k)$ and~$d_i(k)$ or even a return model that takes multiple outcomes, e.g., state prices~\cite{duffie2010dynamic, luenberger2013investment}.
We envision this extension may necessitate the development of innovative mathematical tools to navigate the highly nonlinear, stochastic, and correlated nature of such models.
Additionally, the impact of transaction costs could be considered to embed into the model to enhance the practicality of the proposed trading policy; see \cite{hsieh2022robustness} for an initial result along this line of research. 
Lastly, as the generalized lattice model fits well in a daily basis time scale, it is interesting to see if the model works well in high-frequency settings, such as tick-by-tick basis; see \cite{goldstein2022high}.



\bibliographystyle{informs2014}
\bibliography{references}

%
%
%

\bigskip
 \appendix

\medskip
\section{Technical Proofs of Sections~\ref{section: gain-loss analysis} and~\ref{section: parameters estimation} }\label{Appendix: Technical Proofs of Section: gain-loss analysis}

\medskip
\begin{proof}[Proof of Lemma~\ref{lemma: Survivaliblity}] 
	For $i = 1,2, \dots,n$, Asset~$i$ has per-period returns $X_i(j) \in \{u_i, d_i\}$ with $-1<d_i<0<u_i<1$. 
	Hence, the account value of the long position for the $i$th asset is given by
	\begin{align*}
		V_{i, L}(k) 
		=  \alpha V_{i,0}  R_{i+}(k)  
		& =  \alpha V_{i,0} \prod_{j=0}^{k-1}\left( (1+r_f) + w_i (X_i(j) - r_f)\right)  \\
		& \geq  \alpha V_{i,0} \left( (1+r_f) + w_i (d_i - r_f)\right)^{k}  \\
		& \geq  \alpha V_{i,0} \left( (1+r_f) - w_i (1 + r_f)\right)^{k}  \\
		& =  \alpha V_{i,0} \left( (1+r_f) (1- w_i )\right)^{k} \geq 0,
	\end{align*}
	where the last inequality hold since $w_i \in [0,1]$ and $V_{i,0} \geq 0$ and $r_f \geq 0.$
	Similarly, 
	\begin{align*}
		V_{i, S}(k) &=
		(1-\alpha) V_{i,0} R_{i-}(k) 
		\geq 	(1-\alpha) V_{i,0}\left(1-w_i  \right)^k \geq 0.
	\end{align*}
	Hence, it follows that $V(k) = \sum_{i=1}^n (V_{i,L}(k) + V_{i,S}(k)) \geq 0$ for all $k$ with probability one.
	To complete the proof, we note that if $v_i>0$ for some $i$ and $w_i \in (0,1)$ and~$\alpha \in (0,1)$, then ~$V_{i, L}(k)$ and $V_{i, S}(k)$ are strictly positive. 
	Hence, $V(k) > 0$ for all $k$ with probability one.
\end{proof}

\medskip
\begin{proof}[Proof of Lemma~\ref{lemma: probabilistic positivity}] 
	Fix $\theta \in [0,1]$.
	Note that, by Lemma~\ref{lemma: Survivaliblity}, $V(k) >0$ for all $k$ with probability one. 
	Observe that
	\begin{align*}
		\mathbb{E}\left[ V(k) - \theta \mathbb{E}[V(k)] \right]
		&\leq  \mathbb{E}\left[ (V(k) - \theta \mathbb{E}[V(k)]) \cdot \mathbb{1}_{ \{V(k) > \theta \mathbb{E}[V(k)] \} } \right]\\
		& \leq \mathbb{E}[(V(k) - \theta \mathbb{E}[V(k)])^2 ]^{1/2} \cdot \mathbb{P}( V(k) > \theta \mathbb{E}[V(k)])^{1/2}
	\end{align*}
	where the last inequality holds by the Cauchy-Schwarz inequality. 
	Hence, it follows that
	\begin{align*}
		\mathbb{P}( V(k) > \theta \mathbb{E}[V(k)])
		&\geq \frac{ (\mathbb{E}\left[ V(k) - \theta \mathbb{E}[V(k)] \right])^2 }{\mathbb{E}[(V(k) - \theta \mathbb{E}[V(k)])^2 ] } \\
		&= \frac{ (1-\theta)^2 (\mathbb{E}\left[ V(k)  \right])^2}{\mathbb{E}[(V(k) - \theta \mathbb{E}[V(k)] )^2 ]} . 
	\end{align*}
	Note that the denominator is
	\begin{align*}
		\mathbb{E}[(V(k) - \theta \mathbb{E}[V(k)] )^2 ] 
		& = \mathbb{E}[(V(k)  -\mathbb{E}[V(k)] + \mathbb{E}[V(k)] - \theta \mathbb{E}[V(k)]  )^2 ] \\
		&= {\rm var}(V(k)) + (1-\theta)^2 (\mathbb{E}[V(k)])^2.
	\end{align*}
	Therefore, we have
	\begin{align*}
		\mathbb{P}( V(k) > \theta \mathbb{E}[V(k)])
		&\geq  \frac{ (1-\theta)^2 (\mathbb{E}\left[ V(k)  \right])^2}{ {\rm var}(V(k)) + (1-\theta)^2 (\mathbb{E}[V(k)])^2},
	\end{align*}
	and the proof is complete.
\end{proof}

\medskip
\begin{proof}[Proof of Lemma~\ref{lemma: zero weight gain-loss}] 
	The proof is elementary. Fix $w_i = 0$ for all $i=1,2, \dots,n$.  Note that 
	$V(k) 
	= \sum_{i=1}^{n} V_{i,0}\left(\alpha R_{i+}(k) + (1-\alpha) R_{i-}(k)\right)$
	with $R_{i+}(k) = (1+r_f)^k $ and $R_{i-}(k)= 1$. 
	Using the fact that $\sum_{i=1}^n v_i = 1$,
	it follows that 
	$
	V(k) = \sum_{i=1}^{n} V_{i,0}\left(\alpha (1+r_f)^k  + (1-\alpha) \right) = (\alpha (1+r_f)^k + (1-\alpha)) V_0.
	$
	Hence, the cumulative trading gain-loss function becomes
	$
	\mathcal{G}(k) = V(k) - V_0 = \alpha( (1+r_f)^k - 1) V_0 \geq 0
	$
	for all~$k$ where the last inequality hold since $\alpha \in [0,1]$ and $r_f \geq 0.$
\end{proof}

\medskip
\begin{proof}[Proof of Lemma~\ref{lemma: recursion formula of probability}]
	By the law of total expectation, it can be computed by 
	\begin{align*}
		p_i(k) 
		&= \mathbb{E} \left[ \mathbb{1}_{\{X_i(k) = u_i\}} \right]\\
		&=\mathbb{E} \left[ \mathbb{E} [ \mathbb{1}_{ \{X_{i}(k) =u_i\}} \mid X(k-1: k-m) ] \right]\\
		&= \mathbb{E} \left[ \mathbb{P} \left( X_i(k) = u_i \mid X( k-1: k-m) \right) \right] \\
		&= \mathbb{E} \left[ \Phi_{i,0} + \sum_{j=1}^m \Phi_{i,j} X_i(k-j) + \sum_{\ell=1}^n \Gamma_{i, \ell} X_\ell(k-1) \right] \\
		&= \Phi_{i,0} + \sum_{j=1}^m \Phi_{i,j} \mathbb{E} \left[ X_i(k-j) \right] + \sum_{\ell=1}^n \Gamma_{i, \ell} \mathbb{E} \left[ X_\ell(k-1) \right],
	\end{align*}
	for stage $k = 0, 1, \dots $. 
	Using the fact that $\mathbb{E}[X_i( k-j )] = (u_i - d_i) \cdot p_i( k-j ) + d_i$ and the initial conditions 
	$
	p_{i}(-j) = \frac{x_{i}(-j) - d_i}{u_i - d_i},
	$
	for $j = 1, 2, \dots, m$, we then have the recursion for probability at stage $k = 0, 1, \dots$,
	\begin{align*}
		p_i (k) 
		&= \Phi_{i,0} + \sum_{j=1}^m \Phi_{i,j} [(u_i - d_i) \cdot p_i({k-j}) + d_i]+ \sum_{\ell=1}^n \Gamma_{i, \ell} [(u_\ell - d_\ell) \cdot p_\ell(k - 1) + d_\ell],
	\end{align*}
	which completes the proof.
\end{proof}

\medskip
\begin{proof}[Proof of Lemma~\ref{lemma: the expectation of Binomial distribution}]
	By Lemma~\ref{lemma: recursion formula of probability},
	the per-period returns of Asset $i$ have Bernoulli distribution with time-varying probability $p_i(\cdot)$. Hence, the expectation of $\log( 1 + w_i X_i(j))$ for a single stage is given by 
	$
	\mathbb{E}[\log(1 + w_i X_i (j))] = \log(1+w_i u_i) p_i(j) + \log(1+w_i d_i) (1-p_i(j)).
	$
	Hence, summing up $k$ stages yields
	\begin{align*}
		\sum_{j=0}^{k-1} \mathbb{E}[\log (1 + w_i X_i(j))] 
		& = \sum_{j=0}^{k-1} \log(1 + w_i u_i) p_i(j) + \log (1 + w_i d_i ) (1 - p_i(j)) \\
		& = \log(1 + w_i u_i) \sum_{j=0}^{k-1} p_i(j) + \log(1+w_i d_i) \sum_{j=0}^{k-1} (1 - p_i(j)) \\
		& = \mathbb{E}[H_i(k)] \log (1 + w_i u_i) + (k - \mathbb{E}[H_i(k)]) \log (1 + w_i d_i),
	\end{align*}
	where $\mathbb{E}[H_i(k)]:= \sum_{j=0}^{k-1} p_i(j)$ is the expected number of the return process of Asset~$i$ that receives positive returns up to stage $k-1$.
	An almost identical proof would work for the short side,  i.e.,
	$$
	\sum_{j=0}^{k-1} \mathbb{E}[\log (1 - w_i X_i (j))] = \mathbb{E}[H_i(k)] \log (1-w_i u_i) + (k-\mathbb{E}[H_i(k)]) \log (1 - w_i d_i),  
	$$
	which completes the proof.
\end{proof}

\medskip
\begin{proof}[Proof of Corollary~\ref{corollary: some special cases}]
	To prove part~$(i)$, we fix $k > 1$. Noting that if $r_f >0$, it would increase the $\overline{\mathcal{G}}(k)$; hence, without loss of generality, we may restrict our attention to the case  $r_f= 0$ in the proof that follows.  Begin by considering the case~$\mathbb{E}[H_i(k)] = 0$. 
	We first compute the functions~$\beta_i(\cdot)$ and~$\gamma_i(\cdot)$ in Theorem~\ref{thm: proof of RPE}; i.e., 
	$\beta_i(k) = ( 1 + w_i d_i)^{k}$
	and
	$\gamma_i(k) = ( 1 - w_i d_i )^k $
	Substituting these two quantities to the lower bound for~$\overline{\mathcal{G}}(k)$ in Theorem~\ref{thm: proof of RPE}, we obtain 
	$
	\overline{\mathcal{G}}(k)>  \sum_{i = 1}^n \frac{V_{i, 0}}{2} \left(  (1 + w_i d_i)^{k}  +  (1 - w_i d_i )^k - 2 \right) >0
	$
	where the last inequality holds by the facts that $V_{i,0} >0$ and $(1+z)^k + (1-z)^k > 2$ for all $z \neq 0$ and all $k > 1$.
	An almost identical proof would work for the case of $\mathbb{E}[H_i(k)] = k$. 
	Indeed, in this case, it is readily verified that
	$ 
	\overline{\mathcal{G}}(k) > \sum_{i = 1}^n \frac{V_{i, 0}}{2} \left(   ( 1 + w_i u_i )^{k}  +  (1 - w_i u_i )^k - 2 \right) > 0.
	$ 
	To prove part~$(ii)$, we note that for $\mathbb{E}[H_i(k)] \in (0, k)$ and take~$(1+w_iu_i)^{\mathbb{E}[H_i(k)]} (1+w_i d_i)^{k - \mathbb{E}[H_i(k)]} + (1- w_i u_i)^{\mathbb{E}[H_i(k)]} (1 - w_i d_i)^{k - \mathbb{E}[H_i(k)]} > 2$ for all~$i$. 
	Applying Theorem~\ref{thm: proof of RPE}, it is readily verified  that $\overline{\mathcal{G}}(k) > 0$ for all $k$ immediately.
\end{proof}

\medskip	
\begin{proof}[Proof of Lemma~\ref{lemma: strictly convexity}]  
	The strict convexity is established by invoking the second-order derivative test.  That is,  for any $\varepsilon>0$,
	$
	\frac{\partial^2}{\partial \varepsilon^2} \theta (\varepsilon)  = a^\varepsilon (\log a)^2 + b^\varepsilon (\log b)^2 >0 
	$
	where the last inequality holds since $a,b >0$ with $a,b\neq 1$ and $(\log a)^2$ and $(\log b)^2 >0.$
	To complete the proof, we must find the minimum, call it~$\varepsilon^*$. Indeed, the first-order condition leads to 
	$
	\frac{\partial}{\partial \varepsilon} \theta (\varepsilon)  = a^\varepsilon \log a+ b^\varepsilon \log b \equiv 0.
	$
	That is, $\varepsilon= \varepsilon^*$ is the minimum of $\phi(\varepsilon)$ if and only if $a^{\varepsilon^* } \log a + b^{\varepsilon^* } \log b =0$.
	Moreover, note that if $a>1$ and $b \in (0,1)$, then $a^{\varepsilon} $ is increasing and $b^{\varepsilon}$ is decreasing. This implies that the sum $\phi(\varepsilon) = a^{\varepsilon} + b^{\varepsilon}$ has a unique minimum $\varepsilon^*>0$. Likewise, if $a \in (0,1)$ and $b>1$, then $a^{\varepsilon}$ is decreasing and $b^{\varepsilon}$ is increasing, hence $\phi(\varepsilon)$ again has a unique minimum $\varepsilon^* >0.$
\end{proof}

\medskip
\begin{proof}[Proof of Corollary~\ref{corollary: limitations of the policies}]
	Fix $k > 1$.
	Without loss of generality, we shall give proof for the case $\mathbb{E}[H_i(k)] > k/2$ as the methodology for proving another side follows an almost identical argument. 
	Indeed, for $i \in \{1,2,\dots,n\}$, $\mathbb{E}[H_i(k)] > k/2$ implies that there exists a constant $\varepsilon_i \in (0, k/2)$ such that $\mathbb{E}[H_i(k)] = k/2 + \varepsilon_i$. Now, take~$\alpha = 1/2$. Theorem~\ref{thm: proof of RPE} implies that the expected gain-loss function satisfies
	\begin{align} \label{eq: G_bar with half alpha and u_i = -d_i}
		\overline{\mathcal{G}}(k) 
		& > \sum_{i = 1}^n \frac{V_{i, 0} }{2}\left(   \beta_i(k)  + \gamma_i(k) -2 \right),
	\end{align}
	where the two quantities $\beta_i(k)$ and $\gamma_i(k)$ satisfy
	\begin{align*}
		\beta_i(k) 
		& = (1 + w_i u_i)^{\mathbb{E}[H_i(k)]} (1+ w_i d_i)^{k - \mathbb{E}[H_i(k)]} \\
		& = (1 + w_i u_i)^{ k/2+\varepsilon_i} (1+ w_i d_i)^{k -(k/2 + \varepsilon_i)} \\
		& = [(1 + w_i u_i)(1+ w_i d_i)]^{k/2 } \left( \frac{1 + w_i u_i}{1+ w_i d_i} \right)^{\varepsilon_i}
	\end{align*}
	and
	\begin{align*}
		\gamma_i(k) 
		& = (1 - w_i u_i)^{\mathbb{E}[H_i(k)]} (1- w_i d_i)^{k - \mathbb{E}[H_i(k)]} \\
		& = (1 - w_i u_i)^{ k/2+\varepsilon_i} (1- w_i d_i)^{k -(k/2 + \varepsilon_i)} \\
		& = [(1 - w_i u_i)(1- w_i d_i)]^{k/2 } \left( \frac{1 - w_i u_i}{1- w_i d_i} \right)^{\varepsilon_i}.
	\end{align*}
	Since $u_i = - d_i$ for all $i$, it is readily verified that 
	$
	\beta_i(k) 
	= (1 - w_i^2 d_i^2)^{k/2 } \left( \frac{1 - w_i d_i}{1+ w_i d_i} \right)^{\varepsilon_i} 
	$
	and
	$
	\gamma_i(k) 
	= (1 - w_i^2 d_i^2)^{k/2 } ( \frac{1 + w_i d_i}{1- w_i d_i})^{\varepsilon_i}.
	$
	Substituting these into Inequality~(\ref{eq: G_bar with half alpha and u_i = -d_i}) yields
	\begin{align*}
		\overline{\mathcal{G}}(k) 
		& >  \sum_{i = 1}^n \frac{V_{i, 0} }{2}\left(   [1 - w_i^2 d_i^2]^{k/2 } \left[ \left( \frac{1 - w_i d_i}{1+ w_i d_i} \right)^{\varepsilon_i}   +  \left( \frac{1 - w_i d_i}{1+ w_i d_i} \right)^{-\varepsilon_i} \right] - 2 \right)
	\end{align*}
	which is desired.

	Similarly, for the case $\mathbb{E}[H_i(k)] < k/2$ implies that there exists a constant $\varepsilon_i \in (0, k/2)$ such that $\mathbb{E}[H_i(k)] = k/2 - \varepsilon_i$. Following the similar method, take $\alpha = 1/2$ and $u_i = -d_i$, by Theorem~\ref{thm: proof of RPE}, we have 
	\begin{align*}
		\overline{\mathcal{G}}(k) \geq \sum_{i=1}^{n} \frac{V_{i,0}}{2}(\beta_i(k) + \gamma_i(k) -2),
	\end{align*}
	where $\beta_i(k) = (1-w_i^2d_i^2)^{k/2} \left(\frac{1-w_id_i}{1 + w_id_i}\right)^{-\varepsilon_i}$ and $\gamma_i(k) = (1-w_id_i)^{k/2} \left(\frac{1+w_id_i}{1 - w_id_i}\right)^{-\varepsilon_i}$. It follows that 
	\begin{align*}
		\overline{\mathcal{G}}(k) >  \sum_{i = 1}^n \frac{V_{i, 0} }{2}\left(   [1 - w_i^2 d_i^2]^{k/2 } \left[ \left( \frac{1 - w_i d_i}{1+ w_i d_i} \right)^{\varepsilon_i}   +  \left( \frac{1 - w_i d_i}{1+ w_i d_i} \right)^{-\varepsilon_i} \right] - 2 \right).
	\end{align*}

	To prove part~$(ii)$, take $\varepsilon_i = 0$ in the proof above,  then one immediately obtains
	$	
	\overline{\mathcal{G}}(k) 
	> \sum_{i = 1}^n {V_{i, 0} }\left(   [1 - w_i^2 d_i^2]^{k/2 }  - 1 \right)
	$
	and the proof is complete.
\end{proof}

\medskip
\begin{proof}[Proof of Lemma~\ref{lemma: An Strictly Increasing Auxiliary Function}]
	Fix $z > 0$ with $z \neq 1$.
	An elementary calculus proof would work for verifying the strict increasingness of the function $\phi(\varepsilon)$. That is, $ \phi'(\varepsilon) =  z^{\varepsilon} \log z - z^{-\varepsilon} \log z $. Proceeds with a proof by contradiction by assuming that the function was not strictly increasing. Then we have
	$
	z^{\varepsilon} \log z \leq z^{-\varepsilon} \log z. 
	$
	Equivalently
	\begin{align} \label{ineq: strict increasing by contradiction}
		z^{2 \varepsilon} \log z \leq  \log z.
	\end{align}
	If $z > 1$, then $\log z > 0$. Hence  Inequality~(\ref{ineq: strict increasing by contradiction}) implies that $z^{2 \varepsilon} \leq 1$ for $\varepsilon>0$, which is absurd since it contradicts~$z > 1$.
	On the other hand, if $z \in (0,1)$, then $\log z < 0$. In this case, Inequality~(\ref{ineq: strict increasing by contradiction}) becomes
	$
	z^{2 \varepsilon}  \geq  1
	$ for $\varepsilon>0$,
	which is again absurd. Therefore, $ \phi'(\varepsilon) >0$ for $\varepsilon>0$ and the function $\phi(\varepsilon)$ is a strictly increasing function.
	To complete the proof, invoking a usual limiting argument leads to  $\inf \phi(\varepsilon) = 2.$ 
\end{proof}

\medskip
\begin{proof} [Proof of Lemma~\ref{lemma: markov chain conditional prob model}] 
	For $i, \ell=1,\dots,n$ and $j=1,\dots,m$, we begin by re-expressing the products $\Phi_{i,j} u_i$, $\Phi_{i,j} d_i$, $\Gamma_{i,\ell} u_\ell$, and $\Gamma_{i,\ell} d_\ell$ as 
	\begin{align*}
		\Phi_{i,j} u_i 
		&= \frac{u_i + d_i}{2} \Phi_{i,j} + \frac{u_i - d_i}{2}\Phi_{i,j};\\
		\Phi_{i,j} d_{i} 
		&= \frac{u_i + d_i}{2}\Phi_{i,j} - \frac{u_i -d_i}{2}\Phi_{i,j}; \\
		\Gamma_{i,\ell} u_\ell 
		&=  \frac{u_\ell+ d_\ell}{2} \Gamma_{i,\ell} + \frac{u_\ell - d_\ell}{2} \Gamma_{i,\ell}; \\
		\Gamma_{i,\ell} d_{\ell}
		&= \frac{u_\ell + d_\ell}{2} \Gamma_{i,\ell} - \frac{u_\ell - d_\ell}{2} \Gamma_{i,\ell}.
	\end{align*}
	Hence,  the conditions $0 \leq \Phi_{i,0} + \sum_{j=1}^{m} \Phi_{i,j} x_{i}(k-j) + \sum_{\ell = 1}^{n} \Gamma_{i, \ell} x_{\ell} (k-1)\leq 1$ with $x_{i}(k-j) \in \{u_i, d_i\}$ holds if and only if 
	\begin{align*}
		0 \leq &\Phi_{i, 0} + \left(\frac{u_i + d_i}{2} \sum_{j=1}^{m} \Phi_{i,j} \right) \pm \frac{u_i - d_i}{2}  \Phi_{i, 1} \pm \frac{u_i - d_i}{2}  \Phi_{i, 2} \pm \dots \pm \frac{u_i - d_i}{2}  \Phi_{i, m} \\
		\qquad\quad &+ \left( \sum_{\ell = 1}^{n} \frac{u_\ell + d_\ell}{2} \Gamma_{i, \ell}\right) \pm \frac{u_1 -d_1}{2} \Gamma_{i,1} \pm \frac{u_2 - d_2}{2} \Gamma_{i,2} \pm \dots \pm \frac{u_n - d_n}{2}\Gamma_{i,n} \leq 1.
	\end{align*}
	Therefore, it suffices to ensure that the following inequalities hold 
	\begin{align*}
		&0 \leq \Phi_{i, 0} + \frac{u_i + d_i}{2}\sum_{j=1}^{m} \Phi_{i,j} + \sum_{\ell=1}^{n} \frac{u_\ell + d_\ell}{2}\Gamma_{i,\ell} + \frac{u_i - d_i}{2}\sum_{j=1}^{m} \left|\Phi_{i,j}\right| + \sum_{\ell=1}^{n} \frac{u_\ell - d_\ell}{2} \left|\Gamma_{i,\ell}\right| \leq 1;  \\
		&0 \leq \Phi_{i, 0} + \frac{u_i + d_i}{2}\sum_{j=1}^{m} \Phi_{i,j} +\sum_{\ell=1}^{n} \frac{u_\ell + d_\ell}{2} \Gamma_{i,\ell} - \frac{u_i - d_i}{2}\sum_{j=1}^{m} \left|\Phi_{i,j} \right| - \sum_{\ell=1}^{n} \frac{u_\ell-d_\ell}{2} \left|\Gamma_{i,\ell}\right| \leq 1.
	\end{align*}
	Via a straightforward calculation, it follows that 
	$\frac{u_i - d_i}{2} \sum_{j=1}^{m}\left|\Phi_{i, j}\right| + \sum_{\ell=1}^{n}\frac{u_\ell - d_\ell}{2}\left|\Gamma_{i,\ell}\right| \leq \frac{1}{2}$ and 
	\begin{align*}
		-\frac{1}{2} + \frac{u_i - d_i}{2} \sum_{j =1}^{m}\left|\Phi_{i,j}\right| + \sum_{\ell=1}^{n} \frac{u_\ell - d_\ell}{2} \left|\Gamma_{i,\ell}\right| 
		&\leq \Phi_{i,0} - \frac{1}{2} + \frac{u_i + d_i}{2}\sum_{j=1}^{m} \Phi_{i,j} +\sum_{\ell=1}^{n} \frac{u_\ell + d_\ell}{2}\Gamma_{i,\ell} \\
		&\leq \frac{1}{2} - \frac{u_i - d_i}{2} \sum_{j=1}^{m}\left|\Phi_{i,j}\right| - \sum_{\ell=1}^{n} \frac{u_\ell - d_\ell}{2} \left|\Gamma_{i,\ell}\right|.
	\end{align*}
	These conditions can be written as 
	$$
	\left|\Phi_{i, 0} - \frac{1}{2}+ \frac{u_i + d_i }{2}\sum_{j=1}^{m} \Phi_{i,j} +\sum_{\ell=1}^{n} \frac{u_\ell + d_\ell}{2}\Gamma_{i,\ell}\right| + \frac{u_i - d_i}{2}\sum_{j=1}^{m} \left|\Phi_{i,j}\right| +\sum_{\ell=1}^{n} \frac{u_\ell - d_\ell}{2} \left| \Gamma_{i,\ell}\right|\leq \frac{1}{2}.
	$$
	To complete the proof, we note that the inequalities are linear in $\Phi_{i,j}$ and $\Gamma_{i,\ell}$, which forms a convex polyhedron. 
\end{proof}

\medskip
\section{Estimated Parameters of the S\&P 30 Stocks} \label{appendix: Estimated Parameters of the S&P30 Stocks}
This appendix summarizes the estimated parameters for the S\&P 30 stocks used in Section~\ref{section: empirical studies}. Table~\ref{table: Estimated U and D} include the estimated movement factors $\widehat{U}$ and $\widehat{D}$, Tables~\ref{table: Estimated Phi_i_j when m = 1} to~\ref{table: Estimated Phi_i_j when m = 10} summarize the estimated serial correlation coefficients~$\Phi_{i,j}$ with various memory length $m \in \{1, 2, 5, 10\}$, and Table~\ref{table: Estimated Correlation Matrix} record the asset correlation coefficient matrix~$\Gamma$. In the following tables, the symbol ``mean" refers to the average of the data and``std" refers to the standard deviation of the data.

\begin{table}[h!]
	\footnotesize
	\centering
	\caption{Estimated $\widehat{U}$ and $\widehat{D}$.}
	\label{table: Estimated U and D}
	\begin{tabular}{l|cc}
		& $\widehat{U}$ & $\widehat{D}$ \\
		\hline
		\textbf{AAPL}  & 0.0173          & -0.0175         \\
		\textbf{ABBV}  & 0.0101          & -0.0109         \\
		\textbf{ADBE}  & 0.0207          & -0.0212         \\
		\textbf{AMZN}  & 0.0229          & -0.0242         \\
		\textbf{AVGO}  & 0.0187          & -0.0182         \\
		\textbf{BAC}   & 0.0165          & -0.0146         \\
		\textbf{BRK.B} & 0.0112          & -0.0108         \\
		\textbf{COST}  & 0.0138          & -0.0143         \\
		\textbf{CSCO}  & 0.0136          & -0.0129         \\
		\textbf{CVX}   & 0.0158          & -0.0163         \\
		\textbf{GOOG}  & 0.0183          & -0.0190         \\
		\textbf{GOOGL} & 0.0188          & -0.0188         \\
		\textbf{HD}    & 0.0146          & -0.0153         \\
		\textbf{JNJ}   & 0.0088          & -0.0080         \\
		\textbf{JPM}   & 0.0149          & -0.0142         \\
		\textbf{KO}    & 0.0086          & -0.0099         \\
		\textbf{LLY}   & 0.0141          & -0.0127         \\
		\textbf{MA}    & 0.0156          & -0.0152         \\
		\textbf{MCD}   & 0.0099          & -0.0089         \\
		\textbf{META}  & 0.0249          & -0.0299         \\
		\textbf{MRK}   & 0.0099          & -0.0089         \\
		\textbf{MSFT}  & 0.0173          & -0.0170         \\
		\textbf{NVDA}  & 0.0305          & -0.0333         \\
		\textbf{PEP}   & 0.0088          & -0.0092         \\
		\textbf{PG}    & 0.0101          & -0.0107         \\
		\textbf{TSLA}  & 0.0291          & -0.0349         \\
		\textbf{UNH}   & 0.0108          & -0.0130         \\
		\textbf{V}     & 0.0144          & -0.0143         \\
		\textbf{WMT}   & 0.0106          & -0.0120         \\
		\textbf{XOM}   & 0.0174          & -0.0175   \\
		\hline
		mean           & 0.0156          &  
		-0.0161\\
		std 			& 0.0057          &  
		0.0069\\
		min           & 0.0086          &  
		-0.0349\\
		max           & 0.0305          &  -0.0080
		\\
	\end{tabular}
\end{table}

\begin{table}[h!]
	\footnotesize
	\centering
	\caption{Estimated $\widehat{\Phi}_{i,j}$ for Memory Length $m=1$.}	
	\label{table: Estimated Phi_i_j when m = 1}
	\begin{tabular}{l|cc}
		& {$\widehat{\Phi}_{i,0}$} & {$\widehat{\Phi}_{i,1}$}\\ 
		\hline 
		\textbf{AAPL}  & 0.4743          & -7.0577         \\
		\textbf{ABBV}  & 0.5648          & -3.3926         \\
		\textbf{ADBE}  & 0.4675          & -4.5233         \\
		\textbf{AMZN}  & 0.4637          & -3.1458         \\
		\textbf{AVGO}  & 0.4934          & -7.1075         \\
		\textbf{BAC}   & 0.4425          & -1.7751         \\
		\textbf{BRK.B} & 0.5066          & -8.0508         \\
		\textbf{COST}  & 0.4891          & -5.7026         \\
		\textbf{CSCO}  & 0.4554          & -6.9281         \\
		\textbf{CVX}   & 0.5713          & -0.4806         \\
		\textbf{GOOG}  & 0.4617          & -6.5382         \\
		\textbf{GOOGL} & 0.4505          & -7.0812         \\
		\textbf{HD}    & 0.4898          & -4.4559         \\
		\textbf{JNJ}   & 0.4915          & -4.9624         \\
		\textbf{JPM}   & 0.4778          & -4.6107         \\
		\textbf{KO}    & 0.5666          & -6.6020         \\
		\textbf{LLY}   & 0.5287          & -0.5153         \\
		\textbf{MA}    & 0.5037          & -7.0355         \\
		\textbf{MCD}   & 0.4806          & -10.0977        \\
		\textbf{META}  & 0.4717          & -3.4837         \\
		\textbf{MRK}   & 0.5634          & -2.5778         \\
		\textbf{MSFT}  & 0.4679          & -8.2390         \\
		\textbf{NVDA}  & 0.4914          & -2.0229         \\
		\textbf{PEP}   & 0.5407          & -9.1421         \\
		\textbf{PG}    & 0.5140          & -5.9338         \\
		\textbf{TSLA}  & 0.4901          & -1.3042         \\
		\textbf{UNH}   & 0.5670          & -9.1984         \\
		\textbf{V}     & 0.5042          & -5.5996         \\
		\textbf{WMT}   & 0.5403          & -2.6996         \\
		\textbf{XOM}   & 0.5775          & 0.2417    \\
		\hline
		mean           & 0.5036          &  
		-5.0007\\
		std           & 0.0406          &  
		2.8112\\
		min           &   0.4425        &  
		-10.0977\\
		max           &   0.5775        &  
		0.2417\\     
	\end{tabular}
\end{table}

\begin{table}[!ht]
	\footnotesize
	\centering
	\caption{Estimated $\widehat{\Phi}_{i,j}$ for Memory Length $m=2$.}
	\label{table: Estimated Phi_i_j when m = 2}
	\begin{tabular}{l|ccc}
		& $\widehat{\Phi}_{i,0}$ & $\widehat{\Phi}_{i,1}$ & $\widehat{\Phi}_{i,2}$  \\ \hline
		\textbf{AAPL}   & 0.4724    & -6.8957   & -3.2137 \\ 
		\textbf{ABBV}   & 0.5710    & -3.2523   & -4.3757 \\ 
		\textbf{ADBE}   & 0.4676    & -4.5708   & -1.2987 \\ 
		\textbf{AMZN}   & 0.4637    & -2.9836   & -0.9418 \\ 
		\textbf{AVGO}   & 0.4922    & -7.0590   & -0.2977 \\ 
		\textbf{BAC}    & 0.4406    & -1.9511   & 0.2551 \\ 
		\textbf{BRK.B}  & 0.5054    & -8.0853   & -3.9996 \\ 
		\textbf{COST}   & 0.4909    & -5.8565   & -0.9620 \\ 
		\textbf{CSCO}   & 0.4562    & -7.0756   & -1.4489 \\ 
		\textbf{CVX}    & 0.5732    & -0.4694   & -1.8346 \\ 
		\textbf{GOOG}   & 0.4608    & -6.3603   & -1.6429 \\ 
		\textbf{GOOGL}  & 0.4480    & -6.9173   & -2.4574 \\ 
		\textbf{HD}     & 0.4849    & -4.1834   & -4.1064 \\ 
		\textbf{JNJ}    & 0.4939    & -4.7201   & -1.8473 \\ 
		\textbf{JPM}    & 0.4776    & -4.8785   & 3.6627 \\ 
		\textbf{KO}     & 0.5656    & -6.7373   & -1.6632 \\ 
		\textbf{LLY}    & 0.5304    & -0.7643   & 0.8080 \\ 
		\textbf{MA}     & 0.5017    & -7.1727   & 0.0707 \\ 
		\textbf{MCD}    & 0.4822    & -9.7406   & 2.5554 \\ 
		\textbf{META}   & 0.4649    & -3.4820   & -2.1993 \\ 
		\textbf{MRK}    & 0.5651    & -2.7137   & -1.9016 \\ 
		\textbf{MSFT}   & 0.4675    & -8.3801   & -2.6249 \\ 
		\textbf{NVDA}   & 0.4920    & -1.8762   & -0.6963 \\ 
		\textbf{PEP}    & 0.5398    & -9.0190   & -1.0837 \\ 
		\textbf{PG}     & 0.5130    & -5.7470   & 2.9805 \\ 
		\textbf{TSLA}   & 0.4904    & -1.2043   & -0.5430 \\ 
		\textbf{UNH}    & 0.5691    & -9.0458   & -0.0528 \\ 
		\textbf{V}      & 0.5021    & -5.6236   & -1.9024 \\ 
		\textbf{WMT}    & 0.5427    & -2.9355   & 0.0720 \\ 
		\textbf{XOM}    & 0.5773    & 0.1968    & -0.4443 \\ \hline
		{mean}   & 0.5034    & -4.9835   & -1.0378 \\ 
		{std}   & 0.0417    & 2.7656   & 1.8998 \\ 	
		{min}    & 0.4406    & -9.7406   & -4.3757 \\ 
		{max}    & 0.5773    & 0.1968    & 3.6627 \\ 
	\end{tabular}
\end{table}

\begin{table}[h!]
	\footnotesize
	\centering
	\caption{Estimated $\widehat{\Phi}_{i,j}$ for Memory Length $m=5$.}
	\label{table: Estimated Phi_i_j when m = 5}
	\begin{tabular}{l|cccccc}
		& $\widehat{\Phi}_{i,0}$ & $\widehat{\Phi}_{i,1}$ & $\widehat{\Phi}_{i,2}$ & $\widehat{\Phi}_{i,3}$ & $\widehat{\Phi}_{i,4}$ & $\widehat{\Phi}_{i,5}$ \\
		\hline
		\textbf{AAPL}  & 0.4903   & -6.0457    & -2.5566    & 1.3849     & -0.5490    & 1.6469     \\
		\textbf{ABBV}  & 0.5774   & -3.4784    & -4.6048    & -3.4273    & -0.1246    & -0.1277    \\
		\textbf{ADBE}  & 0.4706   & -4.7356    & -1.6543    & 1.4729     & -1.5489    & -1.0536    \\
		\textbf{AMZN}  & 0.4670   & -3.0835    & -0.9198    & -0.2644    & -2.0204    & 1.5314     \\
		\textbf{AVGO}  & 0.4960   & -7.2420    & -0.4197    & -0.1947    & -0.0650    & -1.7599    \\
		\textbf{BAC}   & 0.4369   & -1.8106    & 0.3323     & -0.0738    & -1.6913    & 1.8038     \\
		\textbf{BRK.B} & 0.4974   & -8.8585    & -4.9791    & 0.6956     & -1.9145    & -2.2998    \\
		\textbf{COST}  & 0.4967   & -6.3076    & -1.1686    & -1.8999    & -2.6737    & 4.0319     \\
		\textbf{CSCO}  & 0.4513   & -7.2763    & -1.2937    & 2.1238     & -1.5100    & -1.8733    \\
		\textbf{CVX}   & 0.5620   & -0.8948    & -2.0258    & -0.1786    & 1.8117     & 1.7782     \\
		\textbf{GOOG}  & 0.4666   & -6.5287    & -1.5938    & -0.3628    & -2.0394    & 0.5425     \\
		\textbf{GOOGL} & 0.4657   & -6.4507    & -1.9203    & 0.0000     & -2.2459    & 0.0000     \\
		\textbf{HD}    & 0.4913   & -4.2705    & -4.1762    & 0.3816     & 0.6885     & 0.5236     \\
		\textbf{JNJ}   & 0.4866   & -4.9292    & -2.8317    & 3.0761     & 3.2063     & 10.3256    \\
		\textbf{JPM}   & 0.4720   & -4.8995    & 4.0073     & 0.3229     & -1.6667    & -1.8565    \\
		\textbf{KO}    & 0.5686   & -7.1570    & -1.5807    & -0.7660    & -1.9007    & -0.2845    \\
		\textbf{LLY}   & 0.5337   & -0.8408    & 0.9450     & -0.6955    & -3.2962    & 2.9291     \\
		\textbf{MA}    & 0.5069   & -7.3545    & 0.1202     & 0.5930     & 0.2353     & 0.3416     \\
		\textbf{MCD}   & 0.4828   & -9.5574    & 2.9993     & -3.3113    & 2.2448     & 6.0217     \\
		\textbf{META}  & 0.4621   & -3.4509    & -2.1904    & -0.3271    & -0.8298    & 0.0918     \\
		\textbf{MRK}   & 0.5669   & -2.1407    & -2.0788    & -2.3531    & 3.5751     & -3.9520    \\
		\textbf{MSFT}  & 0.4835   & -7.6361    & -1.9472    & -0.6301    & 0.0000     & -1.7776    \\
		\textbf{NVDA}  & 0.4882   & -1.9375    & -0.4615    & -1.1027    & -1.0826    & 0.1271     \\
		\textbf{PEP}   & 0.5364   & -10.4331   & -1.7617    & -2.7711    & -5.2843    & -3.3639    \\
		\textbf{PG}    & 0.5139   & -5.5992    & 3.3013     & -2.0523    & 2.6299     & -2.4837    \\
		\textbf{TSLA}  & 0.4904   & -1.3581    & -0.5920    & -0.4209    & 0.4666     & -1.3029    \\
		\textbf{UNH}   & 0.5779   & -9.8296    & -0.6249    & -1.8474    & -0.2913    & -0.2166    \\
		\textbf{V}     & 0.5070   & -5.8935    & -1.8933    & -0.5778    & -1.7258    & 0.2817     \\
		\textbf{WMT}   & 0.5407   & -2.3627    & -0.1161    & 0.7031     & -2.0813    & 0.8319     \\
		\textbf{XOM}   & 0.5644   & -0.1731    & -0.7313    & 0.3908     & 2.1484     & 1.0343     \\
		\hline
		{mean}  & 0.5050   & -5.0845  & -1.0806  & -0.4037  & -0.5845  & 0.3830   \\
		{std}    & 0.0402   & 2.8674  & 2.0398  & 1.5045  & 2.0326  & 2.8278   \\
		{min}   & 0.4369   & -10.4331 & -4.9791  & -3.4273  & -5.2843  & -3.9520  \\
		{max}   & 0.5779   & -0.1731  & 4.0073   & 3.0761   & 3.5751   & 10.3256 
	\end{tabular}
\end{table}

\begin{table}[!ht]
	\footnotesize
	\centering
	\caption{Estimated $\widehat{\Phi}_{i,j}$ for Memory Length $m=10$.}
	\label{table: Estimated Phi_i_j when m = 10}
	\begin{tabular}{l|c c c c c c c c c c c}
		& $\widehat{\Phi}_{i,0}$ & $\widehat{\Phi}_{i,1}$ & $\widehat{\Phi}_{i,2}$ & $\widehat{\Phi}_{i,3}$ & $\widehat{\Phi}_{i,4}$ & $\widehat{\Phi}_{i,5}$ & $\widehat{\Phi}_{i,6}$ & $\widehat{\Phi}_{i,7}$ & $\widehat{\Phi}_{i,8}$ & $\widehat{\Phi}_{i,9}$ & $\widehat{\Phi}_{i,10}$\\
		\hline
		\textbf{AAPL}   & 0.4867    & -6.0468   & -2.3184   & 1.0992    & 0.0000    & 1.1757    & 0.0000    & 0.0000    & 0.0000    & 0.0000    & -1.3424 \\ 
		\textbf{ABBV}   & 0.5722    & -3.0746   & -4.3332   & -3.2077   & -0.2310   & -0.0760   & -0.0477   & 0.2086    & 1.7399    & 0.6545    & -2.9690 \\ 
		\textbf{ADBE}   & 0.4714    & -4.3141   & -1.1528   & 1.1536    & -1.1368   & -0.3080   & 0.0000    & 0.0000    & -1.6589   & -0.8618   & -0.6552 \\ 
		\textbf{AMZN}   & 0.4787    & -2.8012   & -0.4456   & -0.3153   & -1.9531   & 1.3428    & 0.0000    & -0.7664   & -0.8991   & 0.7221    & 1.0201 \\ 
		\textbf{AVGO}   & 0.5082    & -6.5873   & 0.0000    & 0.0000    & 0.0000    & -1.1906   & -3.3023   & 0.0000    & -0.8407   & -0.6345   & -0.6275 \\ 
		\textbf{BAC}    & 0.4393    & -1.3490   & -0.0889   & -0.0362   & -1.2976   & 2.2879    & -2.7961   & 0.6335    & -2.9673   & 1.7658    & 0.1447 \\ 
		\textbf{BRK.B}  & 0.4964    & -8.6050   & -4.4898   & 0.0000    & -1.2611   & -2.0547   & -0.0023   & 1.7707    & -0.7803   & 1.3104    & 1.4179 \\ 
		\textbf{COST}   & 0.5027    & -6.3884   & -0.6677   & -1.7687   & -2.6611   & 3.6459    & 1.0455    & -0.5666   & 1.2112    & 0.5117    & 1.1456 \\ 
		\textbf{CSCO}   & 0.4658    & -7.1298   & -1.0377   & 2.5167    & -0.0741   & -0.5531   & -1.2505   & 2.3559    & -0.0075   & 2.6266    & -0.4315 \\ 
		\textbf{CVX}    & 0.5627    & -0.4194   & -2.1235   & -0.1231   & 1.8365    & 1.8130    & -2.0245   & 2.6408    & -2.5675   & 1.7544    & -0.6850 \\ 
		\textbf{GOOG}   & 0.4672    & -6.1960   & -1.3498   & 0.0000    & -1.7443   & 0.5515    & -0.0591   & 0.0000    & -0.0387   & 0.1627    & 0.9563 \\ 
		\textbf{GOOGL}  & 0.4602    & -6.3070   & -1.8457   & 0.0000    & -2.1501   & 0.0000    & 0.0000    & 0.0000    & 0.0000    & 0.0000    & 0.0190 \\ 
		\textbf{HD}     & 0.5025    & -4.7400   & -4.6531   & -0.3289   & 0.2936    & 0.4028    & -0.5937   & 1.3441    & 1.0331    & 2.0433    & 2.1760 \\ 
		\textbf{JNJ}    & 0.5006    & -5.7181   & -3.9292   & 3.4345    & 2.8180    & 11.8124   & 0.1560    & -0.5695   & -2.5858   & 0.3663    & -5.4969 \\ 
		\textbf{JPM}    & 0.4785    & -4.5622   & 3.2834    & 0.2067    & -0.6562   & -1.5361   & -0.1417   & 0.5352    & 0.7523    & 3.2933    & -2.8566 \\ 
		\textbf{KO}     & 0.5604    & -7.5233   & -1.6110   & -1.0146   & -2.1853   & 0.0297    & -6.1188   & 2.9035    & 0.9533    & 0.0414    & -0.2377 \\ 
		\textbf{LLY}    & 0.5350    & -0.8946   & 1.2286    & 0.2687    & -3.6011   & 3.1966    & 0.0389    & -1.2349   & -3.2674   & 4.8204    & -0.4896 \\ 
		\textbf{MA}     & 0.5014    & -8.0033   & -0.1182   & 0.5164    & 0.9363    & 0.7731    & 0.0000    & -1.4314   & -2.9440   & -0.7316   & 0.3741 \\ 
		\textbf{MCD}    & 0.4984    & -9.4845   & 1.8615    & -2.3841   & 1.0978    & 3.7620    & 3.7078    & 2.8442    & -3.3899   & 0.6593    & 3.6118 \\ 
		\textbf{META}   & 0.4688    & -3.3157   & -1.8679   & -0.5192   & -0.7732   & 0.3499    & -1.1460   & 1.1556    & 0.6323    & 0.0000    & -0.8245 \\ 
		\textbf{MRK}    & 0.5501    & -2.8590   & -1.1644   & -2.1002   & 2.6720    & -3.5841   & 0.2096    & 0.0355    & 6.0822    & 2.8016    & 1.7149 \\ 
		\textbf{MSFT}   & 0.4826    & -7.4980   & -1.6388   & -0.1095   & 0.0000    & -1.1698   & -1.2919   & 0.0000    & 0.0000    & 0.2375    & 0.0000 \\ 
		\textbf{NVDA}   & 0.4831    & -1.7214   & -0.4927   & -1.0516   & -1.1298   & 0.2661    & -0.4853   & 0.8280    & -0.8354   & -0.3877   & 0.3978 \\ 
		\textbf{PEP}    & 0.5298    & -10.4575  & -0.7566   & -1.9033   & -5.1067   & -2.7873   & -3.3101   & 2.0992    & -0.1821   & 0.4434    & -2.7524 \\ 
		\textbf{PG}     & 0.5165    & -5.5333   & 2.5618    & -1.9286   & 3.3486    & -2.7486   & -0.2179   & -0.9519   & -1.0760   & 0.8461    & -0.0148 \\ 
		\textbf{TSLA}   & 0.4973    & -0.8211   & -0.5040   & -0.2528   & 0.6465    & -1.2883   & 0.3078    & 2.3891    & -1.4361   & 1.4047    & -0.8947 \\ 
		\textbf{UNH}    & 0.5718    & -8.7663   & -0.0536   & -0.7555   & 0.0000    & 0.0000    & -1.1147   & 0.1035    & 0.0000    & 1.7695    & -6.8970 \\ 
		\textbf{V}      & 0.5082    & -6.2467   & -1.6439   & -0.7032   & -1.6754   & 0.0739    & -1.0217   & -1.6165   & -0.7638   & -2.4145   & -1.9267 \\ 
		\textbf{WMT}    & 0.5482    & -2.2058   & -0.5936   & 0.9173    & -2.7262   & 0.9771    & -0.7612   & 0.1939    & 1.1571    & -1.2646   & 3.2698 \\ 
		\textbf{XOM}    & 0.5674    & 0.5619    & -0.4573   & 0.3952    & 2.2481    & 1.3658    & -2.9903   & 2.7908    & 0.6023    & 1.5257    & -2.6165 \\ 
		\hline
		{mean}   & 0.5071    & -4.9669   & -1.0134   & -0.2665   & -0.4822   & 0.5510    & -0.7737   & 0.5898    & -0.4026   & 0.7822    & -0.5157 \\ 
		{std} 		& 0.0368	& 2.9305	& 1.8586	& 1.3714	& 1.9466	& 2.7760	& 	1.7149	& 1.3465	& 1.8901	& 1.4626	& 2.2513\\
		{min}    & 0.4393    & -10.4575  & -4.6531   & -3.2077   & -5.1067   & -3.5841   & -6.1188   & -1.6165   & -3.3899   & -2.4145   & -6.8970 \\ 
		{max}    & 0.5722    &  0.5619   & 3.2834    & 3.4345    & 3.3486    & 11.8124   & 3.7078    & 2.9035    & 6.0822    & 4.8204    & 3.6118 \\ 
	\end{tabular}
\end{table}

\begin{table}[]
	\tiny 
	\setlength{\tabcolsep}{1pt}
	\centering
	\caption{Estimated Correlation Matrix $\widehat{\Gamma}$.}
	\label{table: Estimated Correlation Matrix}
	\rotatebox{90}{
		\begin{tabular}{l| llllllllllllllllllllllllllllll} 
			& {AAPL} & {ABBV} & {ADBE} & {AMZN} & {AVGO} & {BAC} & {BRK.B} & {COST} & {CSCO} & {CVX} & {GOOG} & {GOOGL} & {HD} & {JNJ} & {JPM} & {KO} & {LLY} & {MA} & {MCD} & {META} & {MRK} & {MSFT} & {NVDA} & {PEP} & {PG} & {TSLA} & {UNH} & {V} & {WMT} & {XOM} \\
			\hline
			{AAPL}  & 0          & .24          & .71          & .70          & .76          & .57         & .69           & .63          & .64          & .29         & .79          & .80           & .60        & .37         & .55         & .52        & .36         & .75        & .51         & .59          & .28         & .82          & .76          & .55         & .46        & .64          & .49         & .70       & .34         & .28         \\
			{ABBV}  & -          & 0          & .18          & .21          & .26          & .30         & .41           & .30          & .31          & .16         & .22          & .22           & .28        & .51         & .35         & .42        & .51         & .28        & .32         & .10          & .49         & .28          & .19          & .43         & .44        & .06          & .51         & .30       & .25         & .18         \\
			{ADBE}  & -          & -          & 0          & .65          & .71          & .42         & .52           & .54          & .54          & .22         & .72          & .73           & .61        & .20         & .43         & .35        & .30         & .66        & .37         & .60          & .15         & .77          & .72          & .39         & .33        & .50          & .29         & .62       & .27         & .20         \\
			{AMZN}  & -          & -          & -          & 0          & .66          & .55         & .57           & .56          & .48          & .31         & .72          & .72           & .57        & .23         & .50         & .36        & .27         & .60        & .34         & .61          & .16         & .74          & .71          & .36         & .24        & .59          & .32         & .55       & .30         & .25         \\
			{AVGO}  & -          & -          & -          & -          & 0          & .56         & .62           & .57          & .66          & .30         & .73          & .74           & .59        & .25         & .56         & .42        & .32         & .67        & .44         & .56          & .18         & .76          & .83          & .43         & .39        & .59          & .41         & .61       & .22         & .27         \\
			{BAC}   & -          & -          & -          & -          & -          & 0         & .71           & .41          & .50          & .36         & .54          & .54           & .43        & .31         & .90         & .44        & .25         & .60        & .47         & .42          & .26         & .56          & .55          & .40         & .36        & .41          & .42         & .59       & .23         & .31         \\
			{BRK.B} & -          & -          & -          & -          & -          & -         & 0           & .55          & .60          & .47         & .62          & .63           & .57        & .48         & .71         & .62        & .43         & .65        & .50         & .43          & .37         & .64          & .58          & .55         & .46        & .41          & .49         & .62       & .36         & .43         \\
			{COST}  & -          & -          & -          & -          & -          & -         & -           & 0          & .54          & .23         & .57          & .57           & .65        & .41         & .42         & .56        & .36         & .49        & .49         & .37          & .21         & .61          & .56          & .58         & .54        & .46          & .50         & .48       & .59         & .20         \\
			{CSCO}  & -          & -          & -          & -          & -          & -         & -           & -          & 0          & .25         & .57          & .57           & .52        & .41         & .53         & .55        & .43         & .54        & .50         & .42          & .33         & .59          & .55          & .55         & .51        & .36          & .44         & .51       & .36         & .21         \\
			{CVX}   & -          & -          & -          & -          & -          & -         & -           & -          & -          & 0         & .26          & .26           & .17        & .12         & .30         & .19        & .20         & .24        & .13         & .19          & .20         & .26          & .28          & .15         & 0        & .17          & .25         & .21       & .20         & .88         \\
			{GOOG}  & -          & -          & -          & -          & -          & -         & -           & -          & -          & -         & 0          & 1           & .56        & .27         & .51         & .37        & .32         & .66        & .40         & .68          & .23         & .85          & .77          & .41         & .34        & .55          & .39         & .59       & .31         & .21         \\
			{GOOGL} & -          & -          & -          & -          & -          & -         & -           & -          & -          & -         & -          & 0           & .56        & .27         & .51         & .38        & .32         & .66        & .41         & .68          & .23         & .85          & .77          & .42         & .35        & .55          & .40         & .60       & .31         & .22         \\
			{HD}    & -          & -          & -          & -          & -          & -         & -           & -          & -          & -         & -          & -           & 0        & .38         & .48         & .51        & .38         & .57        & .43         & .48          & .25         & .62          & .57          & .52         & .47        & .35          & .39         & .54       & .41         & .14         \\
			{JNJ}   & -          & -          & -          & -          & -          & -         & -           & -          & -          & -        & -          & -           & -        & 0         & .35         & .57        & .60         & .30        & .48         & .16          & .64         & .32          & .18          & .59         & .59        & .10          & .55         & .31       & .33         & .09         \\
			{JPM}   & -          & -          & -         & -          & -          & -         & -           & -          & -          & -         & -          & -           & -        & -         & 0         & .46        & .31         & .60        & .52         & .39          & .29         & .53          & .53          & .43         & .39        & .37          & .44         & .58       & .23         & .28         \\
			{KO}    & -          & -          & -          & -          & -          & -         & -           & -          & -          & -         & -          & -           & -        & -         & -         & 0        & .43         & .50        & .63         & .22          & .42         & .46          & .34          & .84         & .75        & .23          & .54         & .50       & .43         & .18         \\
			{LLY}   & -          & -          & -          & -          & -          & -         & -           & -          & -          & -         & -          & -           & -        & -         & -         & -        & 0         & .38        & .45         & .29          & .59         & .38          & .28          & .47         & .37        & .19          & .56         & .41       & .24         & .18         \\
			{MA}    & -          & -          & -          & -          & -          & -         & -           & -          & -          & -         & -          & -           & -        & -         & -         & -        & -         & 0        & .55         & .53          & .29         & .73          & .68          & .49         & .44        & .47          & .42         & .93       & .26         & .27         \\
			{MCD}   & -          & -          & -          & -          & -          & -         & -           & -          & -          & -         & -          & -           & -        & -         & -         & -        & -         & -        & 0         & .23          & .43         & .47          & .41          & .68         & .57        & .25          & .50         & .53       & .34         & .13         \\
			{META}  & -          & -          & -          & -          & -          & -         & -           & -          & -          & -         & -          & -           & -        & -         & -         & -        & -         & -        & -         & 0          & .22         & .63          & .61          & .28         & .24        & .39          & .18         & .46       & .18         & .14         \\
			{MRK}   & -          & -          & -          & -          & -          & -         & -           & -          & -          & -         & -          & -           & -        & -         & -         & -        & -         & -        & -         & -          & 0         & .25          & .10          & .45         & .45        & .02          & .50         & .30       & .23         & .19         \\
			{MSFT}  & -          & -          & -          & -          & -          & -         & -           & -          & -          & -         & -          & -           & -        & -         & -         & -        & -         & -        & -         & -          & -         & 0          & .79          & .49         & .44        & .55          & .46         & .66       & .32         & .22         \\
			{NVDA}  & -          & -          & -          & -          & -          & -         & -           & -          & -          & -         & -          & -           & -        & -         & -         & -        & -         & -        & -         & -          & -         & -          & 0          & .35         & .29        & .67          & .35         & .62       & .23         & .23         \\
			{PEP}   & -          & -          & -          & -          & -          & -         & -           & -          & -          & -         & -          & -           & -        & -         & -         & -        & -         & -        & -         & -          & -         & -          & -          & 0         & .74        & .22          & .56         & .48       & .47         & .16         \\
			{PG}    & -          & -          & -          & -          & -          & -         & -           & -          & -          & -         & -          & -           & -        & -         & -         & -        & -         & -        & -         & -          & -         & -          & -          & -         & 0        & .15          & .54         & .43       & .42         & .03         \\
			{TSLA}  & -          & -          & -          & -          & -          & -         & -           & -          & -          & -         & -          & -           & -        & -         & -         & -        & -         & -        & -         & -          & -         & -          & -          & -         & -        & 0          & .28         & .44       & .17         & .16         \\
			{UNH}   & -          & -          & -          & -          & -          & -         & -           & -          & -          & -         & -          & -           & -        & -         & -         & -        & -         & -        & -         & -          & -         & -          & -          & -         & -        & -          & 0         & .41       & .33         & .22         \\
			{V}     & -          & -          & -          & -          & -          & -         & -           & -          & -          & -         & -          & -           & -        & -         & -         & -        & -         & -        & -         & -          & -         & -          & -          & -         & -        & -          & -         & 0       & .22         & .24         \\
			{WMT}   & -          & -         & -          & -          & -          & -         & -           & -         & -         & -         & -          & -          & -        & -         & -         & -        & -         & -        & -         & -          & -         & -          & -          & -         & -        & -          & -         & -       & 0         & .17         \\
			{XOM}   & -          & -          & -          & -          & -          & -         & -           & -          & -          & -         & -          & -           & -        & -         & -         & -        & -         & -        & -         & -          & -         & -          & -          & -         & -        & -          & -         & -       & -         & 0        
		\end{tabular}
	}
\end{table}


\end{document}